\documentclass[reqno,12pt]{amsart}

\usepackage[margin=1in]{geometry}
\usepackage{amsmath,amssymb,amsthm,mathtools}
\usepackage[dvipsnames]{xcolor}
\usepackage{enumitem}
\usepackage{tikz}
\usepackage{graphicx}
\colorlet{LinkBlue}{MidnightBlue!99!black}
\usepackage[colorlinks=true,linkcolor=LinkBlue,citecolor=LinkBlue,urlcolor=LinkBlue,hypertexnames=false]{hyperref}
\usepackage{aliascnt}
\usepackage{environ}
\usepackage[nameinlink,noabbrev]{cleveref}
\usetikzlibrary{arrows.meta,calc,patterns}
\usepackage{autonum}
\usepackage{etoolbox}

\makeatletter
\patchcmd{\@maketitle}{\global\topskip42\p@}{\global\topskip14\p@}{}{}
\patchcmd{\@maketitle}{\dimen@34\p@}{\dimen@16\p@}{}{}
\makeatother

\newcommand{\R}{\mathbb{R}}
\newcommand{\N}{\mathbb{N}}

\newcommand{\supp}{\operatorname{supp}}
\newcommand{\C}{\mathbb{C}}
\newcommand{\epsi}{\varepsilon}

\newcommand{\lr}[1]{\left\langle #1 \right\rangle}
\newcommand{\norm}[1]{\lVert #1 \rVert}
\renewcommand{\Re}{\operatorname{Re}}
\renewcommand{\Im}{\operatorname{Im}}
\newcommand{\HH}{\mathcal{H}}

\newcommand{\tOmega}{\widetilde{\Omega}}

\newcommand{\EE}{\mathcal{E}}

\newcommand{\defeq}{\stackrel{\mathrm{def}}{=}}
\newcommand{\1}{\mathbf{1}}

\newcommand{\Di}{\mathcal{D}}
\newcommand{\te}{{\theta}}
\newcommand{\spec}{\operatorname{spec}}

\newcommand{\ran}{\operatorname{ran}}
\newcommand{\loc}{{\operatorname{loc}}}
\newcommand{\comp}{{\operatorname{comp}}}
\newcommand{\Id}{\operatorname{Id}}
\newcommand{\ove}[1]{\overline{#1}}

\usepackage[normalem]{ulem}

\newcommand{\add}[1]{{\color{blue}#1}}

\NewEnviron{equations}{%
\begin{equation}\begin{gathered}
  \BODY
\end{gathered}\end{equation}
}

\theoremstyle{plain}
\newtheorem{theorem}{Theorem}[section]
\newaliascnt{proposition}{theorem}
\newtheorem{proposition}[proposition]{Proposition}
\aliascntresetthe{proposition}
\newaliascnt{lemma}{theorem}
\newtheorem{lemma}[lemma]{Lemma}
\aliascntresetthe{lemma}
\newaliascnt{corollary}{theorem}
\newtheorem{corollary}[corollary]{Corollary}
\aliascntresetthe{corollary}
\theoremstyle{definition}
\newaliascnt{definition}{theorem}
\newtheorem{definition}[definition]{Definition}
\aliascntresetthe{definition}
\newaliascnt{example}{theorem}

\aliascntresetthe{example}
\newaliascnt{assumption}{theorem}

\aliascntresetthe{assumption}
\newaliascnt{assignment}{theorem}

\aliascntresetthe{assignment}
\theoremstyle{remark}
\newaliascnt{remark}{theorem}
\newtheorem{remark}[remark]{Remark}
\aliascntresetthe{remark}

\crefname{theorem}{theorem}{theorems}
\Crefname{theorem}{Theorem}{Theorems}
\crefname{proposition}{proposition}{propositions}
\Crefname{proposition}{Proposition}{Propositions}
\crefname{lemma}{lemma}{lemmas}
\Crefname{lemma}{Lemma}{Lemmas}
\crefname{corollary}{corollary}{corollaries}
\Crefname{corollary}{Corollary}{Corollaries}
\crefname{definition}{definition}{definitions}
\Crefname{definition}{Definition}{Definitions}
\crefname{example}{example}{examples}
\Crefname{example}{Example}{Examples}
\crefname{remark}{remark}{remarks}
\Crefname{remark}{Remark}{Remarks}
\crefname{assumption}{assumption}{assumptions}
\Crefname{assumption}{Assumption}{Assumptions}
\crefname{assignment}{assignment}{assignments}
\Crefname{assignment}{Assignment}{Assignments}
\crefname{equation}{equation}{equations}
\Crefname{equation}{Equation}{Equations}

\numberwithin{equation}{section}

\makeatletter
\def\l@subsection{\@tocline{2}{0pt}{3pc}{5pc}{}}
\makeatother

\title{Reflectionless edge states in Dirac models of topological insulators}
\author{Alexis Drouot}
\address[Alexis Drouot]{University of Washington, Seattle, USA.}
\email{adrouot@uw.edu}

\begin{document}

\begin{abstract}
We study an effective model of topological insulators joined along a bent interface. The Hamiltonian is a Dirac operator  $\Di = D_1\sigma_1 + D_2\sigma_2 + \kappa(x) \sigma_3$ on $L^2(\R^2,\C^2)$, where $\kappa(x)$ represents a domain wall with a corner: while initially straight, it undergoes a bend when passing from one side of $\R^2$ to the other. For most energies in a spectral window centered at $0$, we construct a reflectionless edge state of $\Di$: a distorted plane wave with transmission coefficient of modulus $1$. This implies that $\Di$ admits modes confined near $\kappa^{-1}(\{0\})$, that travel through the bend with no loss neither to the bulk nor to reflection. Our results explain experimental observations from the physics literature -- see e.g. the photonics experiments in \cite{RechtsmanEtAl13,SusstrunkHuber15}. We confirm them with various numerical simulations.
\end{abstract}

\maketitle

\section{Introduction}\label{sec:intro}

Topological insulators are striking materials that do not conduct in their interior, but support currents along edges or interfaces separating distinct topological phases. These currents are remarkable in two respects: they are unidirectional, travelling along the interface in one direction only; and they are robust, in that they are not obstructed by imperfections of the medium. Realizations range from quantum Hall systems and gapped honeycomb structures such as graphene, to their photonic, acoustic and mechanical analogues, where interface transport is now routinely observed.

\begin{figure}[b]
\centering
\includegraphics[width=\textwidth]{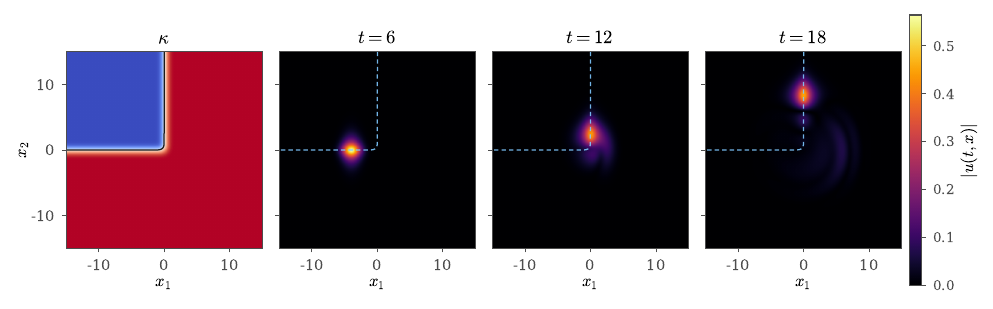}
\caption{The Dirac evolution $\left(D_t + \Di\right)u = 0$ with a corner-like domain wall. There is no visible mass reflected by the corner. See \Cref{sec:numerics2} for further results, and \url{https://sites.google.com/view/alexis-drouot/numerics} for movies of the simulations. }
\label{fig:corner}
\end{figure}

Graphene and its honeycomb analogues provide key backgrounds for the study of topological phases of matter. Near an interface, and after a suitable rescaling, an effective Hamiltonian is a Dirac operator with a domain wall:
\begin{equation}\label{eq:intro-dirac}
    \Di = D_1\sigma_1 + D_2\sigma_2 + \kappa(x)\,\sigma_3, \qquad D_j \defeq -i\partial_{x_j}.
\end{equation}
The emergence of $\Di$ from photonic graphene-based structures was first observed by Haldane and Raghu \cite{HaldaneRaghu08,RaghuHaldane08}, and investigated mathematically in \cite{FLW16a,FLW16b,Drouot19a,Drouot19b,DrouotWeinstein20,HHZ20,AW26}. In \eqref{eq:intro-dirac}, $\sigma_1,\sigma_2,\sigma_3$ are the Pauli matrices and $\kappa$ is a real-valued function whose sign records the topological phase of the medium at $x$; see \cite{Bal19}. The interface is the set where $\kappa$ changes sign. When it is the horizontal line $\R e_1$ -- that is, when $\kappa(x) = k(x_2)$ with $k(\mp\infty) = \pm 1$ -- the operator \eqref{eq:intro-dirac} admits an explicit family of edge states,
\begin{equation}\label{eq:intro-edge}
    \psi_\lambda(x) = e^{i\lambda x_1}\,\phi(x_2), \qquad \Di\,\psi_\lambda = \lambda\,\psi_\lambda, \qquad \lambda \in \R,
\end{equation}
where $\phi$ decays exponentially with $|x_2|$. Such functions are plane waves along the interface, confined transversally to it. Their dispersion relation is linear, so they all move in the same direction at constant speed; superposing them produces wavepackets $a(x_1-t)\,\phi(x_2)$ that glide along the interface without reflection.

The main goal of this paper is to prove that this reflectionless transport subsists even when the domain wall undergoes a bend. This provides a mathematical explanation for the reflectionless transport observed along interfaces between topological structures in photonics, acoustics, electromagnetics and meta-materials experiments, see \Cref{fig:experiments} and e.g. \cite{WangEtAl09,RechtsmanEtAl13,SusstrunkHuber15,LuEtAl17,OrazbayevFleury19,TorresEtAl24}. Namely, we construct reflectionless edge states: distorted plane waves for \eqref{eq:intro-bent-dirac} which remain confined near $\{\kappa = 0\}$, and whose reflection coefficient vanishes. At the dynamical level, these waves are associated to solutions of the Dirac equation that fully transmit through the bend. We support our results with numerical simulations. The main technical tools in our analysis are weighted resolvent estimates, current identities, and scattering theory.

\begin{figure}[t]
\centering
\begin{minipage}[c]{0.47\textwidth}
    \centering
    \includegraphics[width=\linewidth]{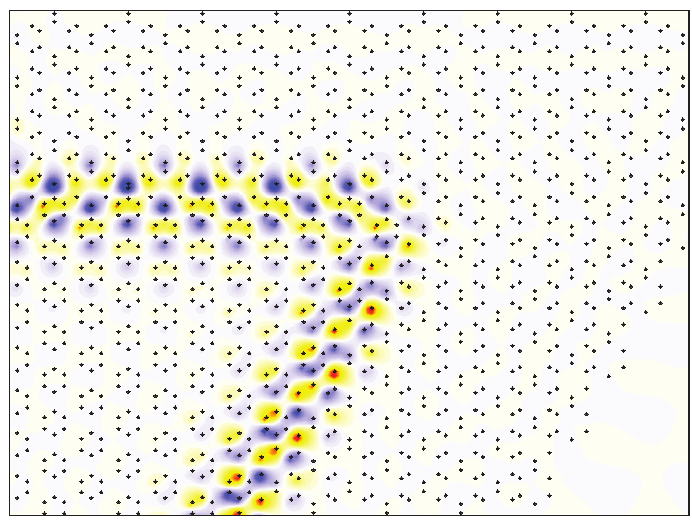}
\end{minipage}\hfill
\begin{minipage}[c]{0.37\textwidth}
    \centering
    \includegraphics[width=\linewidth]{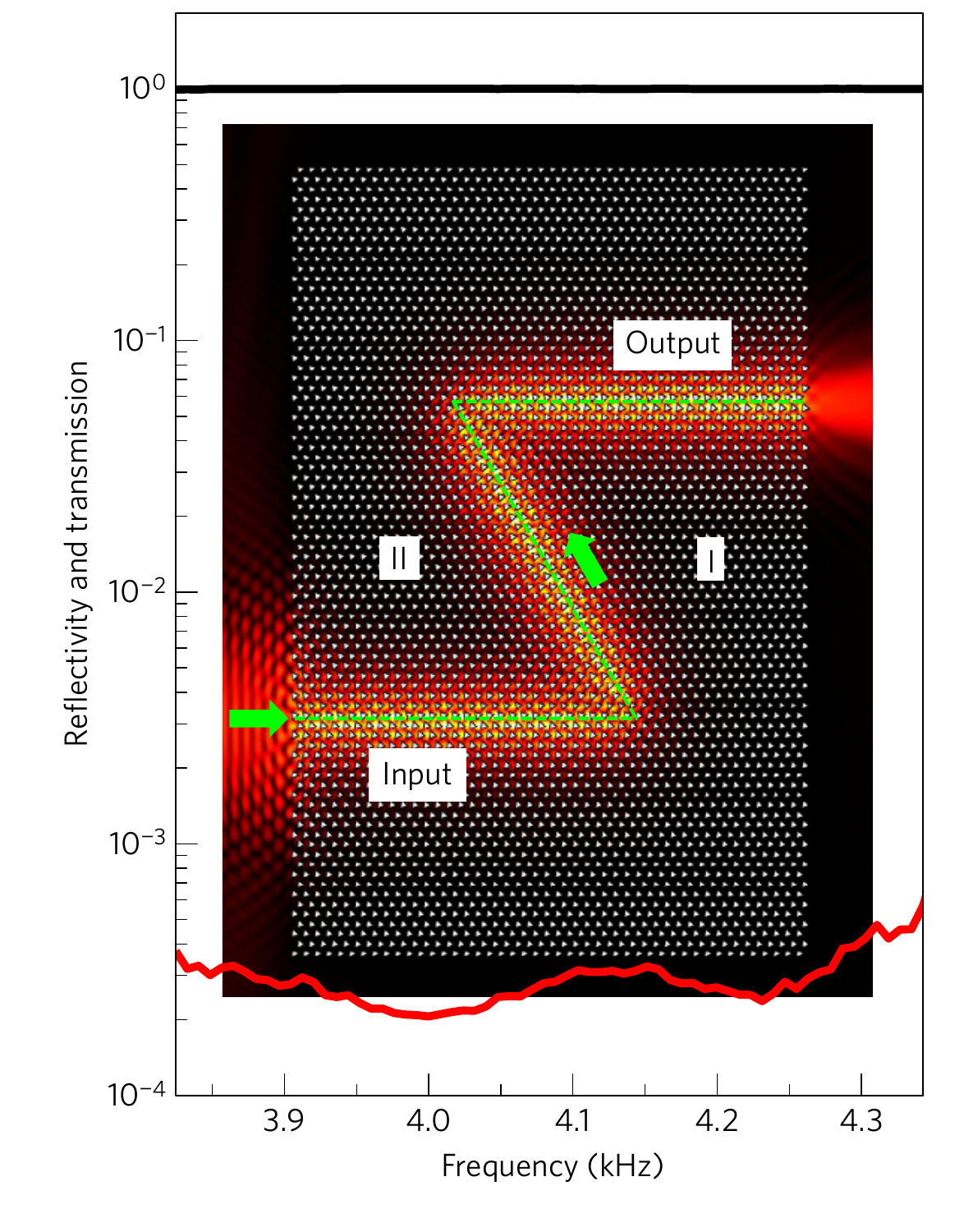}
\end{minipage}
\caption{Edge states experimentally observed in two metamaterials. Left: profile of an edge mode in a graphene-based metamaterial with a $60^\circ$ bend, reproduced from \cite[Figure 3C]{OrazbayevFleury19}. Right: edge mode in a sonic crystal whose interface makes two sharp turns, reproduced from \cite[Figure 4a]{LuEtAl17}, with its transmission (black curve on the top, $\sim 1$) and reflection (red curve on the bottom, $< 10^{-3}$) coefficients.}
\label{fig:experiments}
\end{figure}

\subsection{The model}\label{subsec:intro-model} In this work, we study propagation phenomena for the bent Dirac operator
\begin{equation}\label{eq:intro-bent-dirac}
    \Di \defeq D_1\sigma_1 + D_2\sigma_2 + \kappa(x)\,\sigma_3
\end{equation}
acting on $L^2 = L^2(\R^2,\C^2)$ with domain $H^1 = H^1(\R^2,\C^2)$. In \eqref{eq:intro-bent-dirac}, the matrices $\sigma_1, \sigma_2, \sigma_3$ are the three Pauli matrices
\begin{equation}\label{eq:intro-pauli}
    \sigma_1 \defeq \begin{bmatrix} 0 & 1 \\ 1 & 0 \end{bmatrix}, \qquad
    \sigma_2 \defeq \begin{bmatrix} 0 & -i \\ i & 0 \end{bmatrix}, \qquad
    \sigma_3 \defeq \begin{bmatrix} 1 & 0 \\ 0 & -1 \end{bmatrix}.
\end{equation}
The function $\kappa \in L^\infty(\R^2,\R)$ represents a bent domain wall, see \Cref{fig:intro-bent}. Specifically, for some $\theta_+ \in (-\pi,\pi)$ and some $R > 0$, it satisfies outside of the ball of radius $R$:
\begin{equations}
    \label{eq:intro-bent-minus}
    \kappa(x) = k_-(x \cdot e_-^\perp), \qquad x \in X_- \defeq \left\{ \arg (x) - \dfrac{\theta_+}{2} \in \left[\dfrac{\pi}{2}, \dfrac{3\pi}{2} \right] \right\}, \qquad e_- \defeq \begin{bmatrix} 1 \\ 0\end{bmatrix}
\\
    \kappa(x) = k_+(x \cdot e_+^\perp), \qquad x \in X_+ \defeq \left\{ \arg (x) - \dfrac{\theta_+}{2} \in \left(\dfrac{-\pi}{2}, \dfrac{\pi}{2}\right) \right\}, \qquad e_+ \defeq \begin{bmatrix} \cos \theta_+ \\ \sin \theta_+\end{bmatrix},
\end{equations}
where $k_\pm \in L^\infty(\R,\R)$ are two functions such that:
\begin{equation}\label{eq:intro-profiles}
    k_\pm(s) = 1 \ \text{ for } s \le -1, \qquad k_\pm(s) = -1 \ \text{ for } s \ge 1 .
\end{equation}

\begin{figure}[ht]
\centering
\begin{minipage}[c]{0.54\textwidth}
    \centering
    \includegraphics[width=\linewidth]{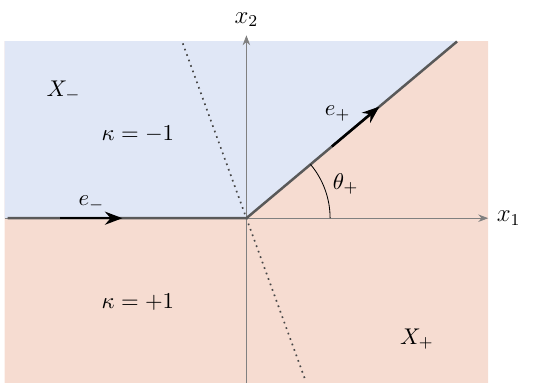}
\end{minipage}\hfill
\begin{minipage}[c]{0.43\textwidth}
\caption{An example of a bent domain wall $\kappa$: it is $-1$ in the blue sector and $+1$ in the red sector. The dotted line separates the two conical regions $X_-$ and $X_+$.}
\label{fig:intro-bent}
\end{minipage}
\end{figure}

Because of the asymptotics \eqref{eq:intro-bent-minus}-\eqref{eq:intro-profiles}, the operator $\Di$ models the junction of two distinct topological phases, represented by
\begin{equation}\label{eq:intro-phases}
    D_1\sigma_1 + D_2\sigma_2 + \sigma_3 \qquad\text{and}\qquad D_1\sigma_1 + D_2\sigma_2 - \sigma_3 ,
\end{equation}
along the (potentially blurred) interface $\R_- e_1 \cup \R_+ e_+$. The two operators \eqref{eq:intro-phases} have a spectral gap $(-1,1)$, and opposite topological indices below the gap \cite{Bal19,Bal22}.\footnote{This also explains why we must require $\theta_+ \neq \pm \pi$: in this case, $\Di$ would be a relatively compact perturbation of $\Di_\pm$. In particular, it would have the essential spectral gap $(-1,1)$, so it cannot have distorted plane waves with energy in $(-1,1)$.} Because of the bulk-edge correspondence, one expects unidirectional currents to emerge along this interface. These are the main objects of study in this work.

\subsection{Main result} 

In the half planes $X_\pm$ of \eqref{eq:intro-bent-minus}, the spatial asymptotics of $\Di$ are described by the operators
\begin{equation}\label{eq:intro-models}
    \Di_\pm \defeq D_1\sigma_1 + D_2\sigma_2 + \kappa_\pm(x)\,\sigma_3, \qquad \kappa_\pm(x) \defeq k_\pm\left(x \cdot e_\pm^\perp\right).
\end{equation}
The operators $\Di_\pm$ admit exact edge states: solutions to $\Di_\pm \psi_{\lambda,\pm} = \lambda\,\psi_{\lambda,\pm}$ propagating along their respective interfaces $\R e_\pm$ and localized transversally to them:
\begin{equation}\label{eq:intro-edge-pm}
    \psi_{\lambda,\pm}(x) \defeq e^{i\lambda\,x\cdot e_\pm}\,\phi_\pm\left(x\cdot e_\pm^\perp\right),
    \qquad
    \phi_\pm(s) \defeq c_\pm \begin{pmatrix} e^{-i\theta_\pm/2} \\ e^{i\theta_\pm/2}\end{pmatrix} e^{\int_0^s k_\pm},
    \qquad \lambda \in \R .
\end{equation}
Here $\theta_- \defeq 0$, and $c_\pm > 0$ normalizes $\phi_\pm$ in $L^2(\R,\C^2)$; the profiles $\phi_\pm$ decay exponentially because of \eqref{eq:intro-profiles}. These represent reflectionless plane waves propagating in the direction of $e_\pm$, confined near $\R e_\pm$.

In this work, we investigate whether Dirac operators with bent domain walls admit generalized eigenstates with similar properties. Formally, our main result states that for most energies in a window $(-E,E)$ (where $E$ has some spectral interpretation), $\Di$ admits \textit{reflectionless} distorted plane waves. It explains the robust, unidirectional transport that has been observed in various hybrid topological structures: see for instance the photonics, electromagnetics and meta-materials experiments realized in \cite{WangEtAl09,RechtsmanEtAl13,SusstrunkHuber15,LuEtAl17,OrazbayevFleury19,TorresEtAl24}.

To state our main result, we introduce:
\begin{itemize}
    \item The space $\EE_\nu \defeq e^{-\nu|x|}L^2$;
    \item A cutoff function $\chi$ such that $\chi(s) = 0$ for $s \leq -1$ and $\chi(s) = 1$ for $s \geq 1$;
    \item The quantities:
\begin{equations}\label{eq:intro-E}
    E \defeq \min\left(E_-,E_+\right), \qquad E_\pm \defeq \min\Big(\spec\big(\Di_\perp^\pm\big)\cap (0,+\infty) \Big),
    \qquad
    \Di_\perp^\pm \defeq D_2\sigma_2 + k_\pm(x_2)\,\sigma_3
\end{equations}
\end{itemize}

\begin{theorem}\label{thm:intro} There exists a set $S \subset (-E,E)$, locally finite in $(-E,E)$, such that for $\lambda \in (-E,E) \setminus S$, the equation $(\Di-\lambda) u = 0$ has a solution $u \in H^1_\loc$ of the form
\begin{equation}\label{eq:intro-main}
    u(x) \; = \; \chi(-x \cdot e_-)\psi_{\lambda,-}(x) \; + \; t(\lambda)\cdot \chi(x \cdot e_+)\psi_{\lambda,+}(x) \; + \; \EE_\nu,
\end{equation}
where $\nu > 0$ and $t(\lambda) \in \C$ has modulus one. 
\end{theorem}

It turns out that the function $u$ constructed in the proof of \Cref{thm:intro} is the unique \textit{distorted plane wave} of $\Di$ at frequency $\lambda$ -- see \Cref{sec:scattering}. We interpret $u$ as a reflectionless edge state of $\Di$. When $x \cdot e_- \ll -1$ -- i.e. $x_1 \ll -1$ -- $u$ looks like $\psi_{\lambda,-}$, i.e. it is a plane wave propagating in the direction of $e_-$, confined near $\R_- e_-$. After it passes the corner, i.e. when $x \cdot e_+ \gg 1$, $u$ becomes a plane wave propagating in the direction of $e_+$, confined near $\R_+ e_+$, \textit{perfectly transmitted:} $|t(\lambda)| = 1$. This is striking because one expects Dirac equations to behave much like wave equations, and have in general non-zero transmission \textit{and} reflection coefficients -- see e.g. \cite{K09} for one-dimensional Dirac systems. It also contrasts with waveguides, where both reflection and transmission occur. For waveguides, a transmission coefficient $|t(\lambda)| = 1$ has been dubbed invisibility \cite{BN13,BCN18,CN18,CP18}.

At the dynamical level, \Cref{thm:intro} means that wavepackets spectrally concentrated (with respect to $\Di$) near a generic energy $\lambda \in (-E,E)$ propagate through the corner without reflection. We refer to \Cref{sec:numerics2} for an informal explanation; \Cref{fig:corner} for preliminary numerical evidence and \Cref{subsec:numerics-results2} for further simulations.

Since $\left|t(\lambda)\right| = 1$, all the transmission information sits in the phase of $t(\lambda)$. The dynamical implication of this phase is a delay. Write $t(\lambda) = e^{i\vartheta(\lambda)}$ with $\vartheta$ real, and let $u_t$ solve $\left(D_t+\Di\right)u_t = 0$, with $u_0$ a wavepacket built from the incoming states $\psi_{\lambda,-}$ with spectral weight concentrated near some $\lambda_0 \in (-E,E)$. Before the corner, $u_t$ is centered on $x\cdot e_- = t$; a stationary phase computation in \eqref{eq:intro-main} places the transmitted packet on $x\cdot e_+ = t - \vartheta'(\lambda_0)$. While not reflecting, the corner is thus not totally invisible to an incoming wavepacket spectrally supported near energy $\lambda_0$: it induces a time shift by
\begin{equation}\label{eq:intro-delay}
    \vartheta'(\lambda_0) = -i\,\ove{t(\lambda_0)}\,\partial_\lambda t(\lambda_0) .
\end{equation}
From a general point of view, this time-shift is large when $\Di$ has resonances near $\lambda_0$, as can be seen from the Breit--Wigner formula \cite[Theorem 2.20]{DZ}. Other characteristics of the propagation at energy $\lambda_0$ can be read off $\vartheta(\lambda)$: for instance the dispersion scale along the interface $\vartheta''(\lambda_0)$.

\subsection{Strategy of proof}\label{subsec:intro-strategy} Our strategy is based on a key object from scattering theory: the meromorphic continuation of the resolvent $(\Di-\lambda)^{-1}$. 

In \Cref{sec:unperturbed} and \Cref{sec:resolvent},  we analyze the resolvent of the reference operator $\Di_*$ attached to a horizontal domain wall $\kappa_*(x) = k_*(x_2)$. The key object is the edge projection $P_*$ onto the space
\begin{equation}\label{eq:intro-edge-sector}
    \HH_* \defeq \left\{ a(x_1)\,\phi_*(x_2) \ : \ a \in L^2(\R) \right\}.
\end{equation}
This projection splits $L^2$ as $\HH_* \oplus \HH_*^\perp$. The operator $\Di_*$ on $\HH_*$ is unitarily equivalent to $D_1$ on $L^2(\R)$. On the other hand, its resolvent $(\Di_*-\lambda)^{-1}$ on $\HH_*^\perp$ is well-defined for $\lambda \in (-E_*,E_*)$, where
\begin{equation}
    E_* \defeq \min\Big(\spec\left(\Di_\perp\right)\cap(0,+\infty)\Big), \qquad \Di_\perp \defeq D_2\sigma_2 + k_*(x_2)\,\sigma_3.
\end{equation}
These observations produce an expansion of the full resolvent of $\Di_*$ on $L^2$ for energies in $(-E_*,E_*)$: for $f \in \EE_\gamma$, we have
\begin{equations}\label{eq:intro-lap}
    \lim_{\epsi \rightarrow 0^+} (\Di_*-\lambda-i\epsi)^{-1} f \; = \; i\left(\int_{-\infty}^{x_1} e^{i\lambda\left(x_1-s\right)}\,a_*(s)\,ds\right)\phi_*(x_2) \; + \; \EE_{\gamma/2} ,
    \\
    P_* f(x) = a_*(x_1)\,\phi_*(x_2) .
\end{equations}
The edge term carries all the long-range behavior: it is a plane wave in $x_1$, confined near the interface, and built from the values of $f$ to the \emph{left} of $x_1$ only. This is where unidirectionality enters the analysis. See \Cref{thm:resolvent-structure} for a more precise statement. In \Cref{sec:tilted} we convert the results of \Cref{sec:unperturbed}--\Cref{sec:resolvent} to tilted analogues $\Di_e$ of $\Di_*$ in the direction $e$.

In order to study the resolvent of $\Di$ itself, in \Cref{sec:bent}, we must face the following difficulty. The operator $\Di$ agrees with $\Di_-$ on $X_-$ and with $\Di_+$ on $X_+$, but it is a compact perturbation of neither: $\kappa - \kappa_\mp$ does not decay on $X_\pm$. We assemble a parametrix from both,
\begin{equation}\label{eq:intro-parametrix}
    Q(\lambda) \defeq \chi_-\,R_-(\lambda) \; + \; \chi_+\,R_+(\lambda), 
\end{equation}
where $\chi_\pm$ is a partition of unity subordinate to $X_\pm$. We show that $\left(\Di-\lambda\right)Q(\lambda) = \Id + K(\lambda)$ for a compact operator $K(\lambda)$. Analytic Fredholm theory \cite[Appendix C]{DZ} then continues
\begin{equation}\label{eq:intro-resolvent-factorization}
    R(\lambda) = Q(\lambda)\left(\Id + K(\lambda)\right)^{-1}
\end{equation}
meromorphically across $(-E,E)$; see \Cref{thm:mero}. Besides continuing the resolvent, \eqref{eq:intro-resolvent-factorization} yields spatial asymptotics of $R(\lambda) f$ from those of $R_\pm(\lambda)f$: see \Cref{lem:outgoing-far-field}. 

We then use the resolvent to construct solutions of
\begin{equation}\label{eq:intro-scattering}
    \left(\Di-\lambda\right)u = 0, \qquad u(x) \sim \psi_{\lambda,-}(x) \ \text{ as } \ x_1 \rightarrow -\infty .
\end{equation}
These take the form
\begin{equation}\label{eq:intro-scattering-state}
    u \; = \; \chi\left(-x\cdot e_-\right)\psi_{\lambda,-} \; + \; R(\lambda)f,
    \qquad
    f \; \defeq \; -\left(\Di-\lambda\right)\left(\chi\left(-x\cdot e_-\right)\psi_{\lambda,-}\right) ,
\end{equation}
where $f$ decays exponentially because $\psi_{\lambda,-}$ is confined near $\R e_-$; see \Cref{lem:source}.
The spatial asymptotics for $R(\lambda)$ imply that $u$ has no reflected component. To justify that all the mass of $u$ is transmitted through the corner or defect, we use a current conservation argument. This will complete the proof of \Cref{thm:intro}. 

\subsection{Relation to existing work}\label{subsec:intro-related} We drew our inspiration and motivation from a recent work from Ammari and Qiu \cite{AmmariQiu26}, who studied edge modes for a graphene-based hybrid structure with an interface bent by the angle $2\pi/3$. They prove that the interface modes persist at generic frequency in the bulk gap. Because the present work directly works with the effective Hamiltonian of such systems, we have obtained more precise results, valid for more general situations. We handle a bend angle that can take an arbitrary value in $(-\pi, \pi)$; and the domain wall may undergo an arbitrary deformation within a bounded region. Our main result, \Cref{thm:intro}, goes beyond producing a mode that survives the bend; it shows that it is fully transmitted across the defect. It would be interesting to investigate whether these conclusions subsist in the Ammari--Qiu setting.

A second body of work studies transport by a semiclassical analogue of the operator $\Di$ \cite{BalEtAl23,BBD24,Bal24,Drouot22,HXZ}. In \cite{BalEtAl23}, wavepackets are propagated along curved interfaces; \cite{BBD24} adds a magnetic field, under which they slow down and disperse; \cite{Bal24} treats dispersive alongside relativistic wavepackets, and analyzes turning points; and \cite{Drouot22} gives a phase-space reduction producing the speed and the profile of the traveling mode for a general class of such operators. Each of these constructions is valid for long yet finite times. The present work is complementary. Using instead a scattering approach and exploiting the asymptotically straight structure of the domain wall, it allows us to bypass the semiclassical result and obtain spectral -- global --  results. Because our results are not asymptotic (there is no small parameter), they are less precise than those of \cite{BalEtAl23,BBD24,Bal24,Drouot22}: for instance, they do not come with the propagation delay induced by the corner.  

The work closest to the present paper is due to Chen and Bal \cite{ChenBal25}. They establish a limiting absorption principle and a generalized eigenfunction expansion for the straight domain wall Dirac operator $\Di_*$, perturbed by a decaying potential. They identify the asymmetry of the resulting transport as a difference of transmission coefficients. The work \cite{ChenBal25} also features weighted resolvent estimates and current conservation arguments. In comparison, the present work also handles bends in the interface, for which anisotropic weighted estimates (see e.g. \Cref{prop:agmon}) and a gluing resolvent construction are necessary. 

\subsection{Further investigations}\label{subsec:intro-further} We propose potential projects on topics related to the present paper, for which the approach to \Cref{thm:intro} could provide some insights. 

One may consider junctions of more than two topological phases. The simplest example is a crossroad
\begin{equation}\label{eq:intro-crossroad}
    \kappa(x) = \operatorname{sgn}\left(x_1\right)\operatorname{sgn}\left(x_2\right).
\end{equation}
The interface is the union of the $x$ and $y$ axes; see the first panel of \Cref{fig:crossroad}. Heuristically, the $x$-axis carries incoming edge states and the $y$-axis carries outgoing edge states. Therefore, we expect that for generic $\lambda$ near $0$,  $\left(\Di-\lambda\right)u=0$ has solutions of the form
\begin{equation}\label{eq:intro-crossroad-state}
    u \; = \; \chi\left(-x_1\right)\psi_{\lambda,-} \; + \; t_\uparrow(\lambda)\cdot\chi\left(x_2\right)\psi_{\lambda,e_2} \; + \; t_\downarrow(\lambda)\cdot\chi\left(-x_2\right)\psi_{\lambda,-e_2} \; + \; \EE_\nu ,
\end{equation}
where $\psi_{\lambda,\pm e_2}$ are the edge states \eqref{eq:intro-edge-pm} attached to the branches $\R_\pm e_2$. Conservation of current should yield $\left|t_\uparrow(\lambda)\right|^2 + \left|t_\downarrow(\lambda)\right|^2 = 1$ in place of $\left|t(\lambda)\right| = 1$. See \Cref{fig:crossroad} for a numerical simulation.

\begin{figure}[t]
\centering
\includegraphics[width=\textwidth]{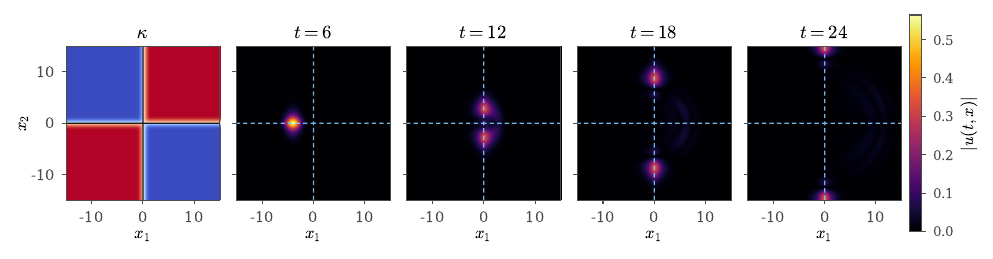}
\caption{The Dirac evolution $\left(D_t+\Di\right)u = 0$ for four alternating topological phases filling quadrants. First panel: the mass $\kappa$, equal to $+1$ in the red quadrants and $-1$ in the blue ones. Remaining panels: $\left|u(t,x)\right|$, with the interface dashed. The packet arrives along $\R_-e_1$ and scatters along two outgoing branches $\R_\pm e_2$. No mass leaves along $\R_+ e_1$ nor reflects along $\R_- e_1$. }
\label{fig:crossroad}
\end{figure}

\Cref{thm:intro} holds for every bend angle $\theta_+ \in (-\pi,\pi)$, with a transmission of modulus one however sharp the turn. When $\theta_+ \rightarrow \pm \pi$, many constants involved in the proof degenerate. This is expected as the operator $\Di$ formally converges to the gapped operator
\begin{equation}
    \sigma_1 D_1 + \sigma_2 D_2 \pm \sigma_3.
\end{equation}
It would be interesting to study the limit $\vartheta(\lambda)$ in the regime $\theta_+ \rightarrow \pm \pi$. In the context of plane waveguides, the sharp-bend limit has been investigated in \cite{DR12}; there the authors show that the number of eigenvalues of the Dirichlet Laplacian blows up as the bend angle goes to $\pm \pi$. 

Another interesting scenario consists of having one of the two topological phases within a large, bounded region: take for instance
\begin{equation}\label{eq:intro-circle}
    \kappa = 1 - 2\cdot\1_{B_R(0)}, \qquad R \gg 1.
\end{equation}
The operator $\Di$ is then a compact perturbation of $\sigma_1 D_1 + \sigma_2 D_2 + \sigma_3$, which has spectral gap $(-1,1)$. Therefore, the spectrum of $\Di$ in $(-1,1)$ consists only of eigenvalues. However, we expect that as $R \rightarrow \infty$, the spectrum of $\Di$ fills $(-1,1)$ -- see e.g. \cite[Theorem 2]{DrouotZhu24} for a related result on discrete systems. Moreover, the spectrum of $\Di$ in $(-1,1)$ should be associated to some form of transport for large $R$, despite it being made of eigenvalues -- see \Cref{fig:circle}. Therefore, two questions have emerged:
\begin{itemize}
    \item What is the distribution of eigenvalues of $\Di$ in $(-1,1)$?
    \item How do eigenvalues of $\Di$ in $(-1,1)$ contribute to dynamics?
\end{itemize} 
We will respond to these questions in a future work.

\begin{figure}[!ht]
\centering
\includegraphics[width=\textwidth]{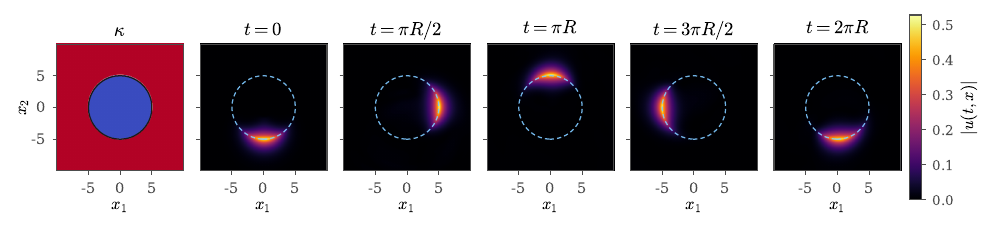}
\caption{The Dirac evolution $\left(D_t+\Di\right)u = 0$ for the circle domain wall \eqref{eq:intro-circle}, with $R = 5$. First panel: the mass $\kappa$, equal to $-1$ in the blue disk $B_R(0)$ and $+1$ in the red exterior. Remaining panels: $\left|u(t,x)\right|$ at different times. The spectrum of $\Di$ in $(-1,1)$ is discrete, nonetheless transport effectively arises.}
\label{fig:circle}
\end{figure}

Some recent papers have studied the bulk-edge correspondence for Dirac operators on a halfplane \cite{BarbarouxEtAl24,BarbarouxEtAl25}, with non-zigzag selfadjoint boundary conditions -- as well as its violation in the zigzag case. It would be interesting to adapt the techniques developed in \cite{DrouotZhu26} to extend the works \cite{BarbarouxEtAl24,BarbarouxEtAl25} to non-straight edges, and the techniques developed in the present work to construct reflectionless edge states for Dirac operators with boundary.

\subsection{Notations}\label{subsec:intro-notations} We will use the following notations:
\begin{itemize}
    \item Function spaces over $\R^2$ with values in $\C^2$ are written $L^2$, $H^1$, $L^2_\loc$, $H^1_\loc$ and $L^2_\comp$. Spaces over $\R$, such as $L^2(\R,\C^2)$, are spelled out.
    \item $\EE_\gamma$ denotes the space $e^{-\gamma|x|}L^2$, normed by $\norm{u}_{\EE_\gamma} \defeq \norm{e^{\gamma|x|}u}_{L^2}$. We write $f = g + \EE_\gamma$ if $f-g \in \EE_\gamma$.
    \item $D_j \defeq -i\partial_{x_j}$ and $D_t \defeq -i\partial_t$.
    \item We write $\1_S$ for the indicator function of a set $S$.
    \item $B_R(0)$ is the ball of radius $R$ centered at $0$.
    \item $\sigma_1, \sigma_2, \sigma_3$ are the Pauli matrices \eqref{eq:intro-pauli}; we set $\sigma \defeq \left(\sigma_1,\sigma_2\right)$ and $\sigma\cdot e \defeq \cos\theta\,\sigma_1 + \sin\theta\,\sigma_2$ for a unit vector $e = \left(\cos\theta,\sin\theta\right)$.
    \item Given a unit vector $e \in \R^2$, $e^\perp$ is its counter-clockwise rotation by $\pi/2$ and $O_e \in SO(2)$ is the rotation mapping $e_1$ to $e$.
    \item $e_- = e_1$ and $e_+$ are the incoming and outgoing branches of the interface, see \Cref{fig:intro-bent}; $\theta_- = 0$ and $\theta_+ \in (-\pi,\pi)$ are the corresponding angles with the $x$-axis (in particular, $\te_-=0$).
    \item $\theta_0 \defeq \tfrac{\theta_+}{2}$ and $e_0$ is the unit vector with angle $\te_0$ with the $x$-axis.
    \item $\kappa$ is a bent domain wall with asymptotics $\kappa_\pm(x) = k_\pm\left(x\cdot e_\pm^\perp\right)$, see \Cref{def:bent}.
    \item  $\Di_*$ is the unperturbed operator \eqref{eq:dirac} and $\Di_e$ is its tilt \eqref{eq:dirac-e}. $\Di_\perp \defeq D_2\sigma_2 + k_*(x_2)\sigma_3$, $\Di_\perp^\pm = D_2\sigma_2 + k_\pm(x_2)\sigma_3$ denote the transverse operators associated to $\Di_*$, $\Di_{e_\pm}$.
    \item $\Di$ is the bent Dirac operator $\sigma_1 D_1 + \sigma_2 D_2 + \kappa(x) \sigma_3$; $\Di_\pm = \Di_{e_\pm}$ denotes its models at infinity, see \eqref{eq:models}.
    \item $E_*$ (resp. $E_\pm$) is the infimum of the positive part of the spectrum of $\Di_\perp$ (resp. $\Di_\perp^\pm$). $E$ is $\min\left(E_-,E_+\right)$.
    \item $\phi_*$ spans the kernel of $\Di_\perp$ -- see \eqref{eq:phi0}; $\phi_e$ and $\phi_\pm = \phi_{e_\pm}$ are its tilted analogues \eqref{eq:psi-lambda-e}, \eqref{eq:phi-pm}, and $M_* \defeq \norm{e^{|x_2|}\phi_*}_\infty$; and $M_\pm \defeq \norm{e^{|x_2|}\phi_\pm}_\infty$.
    \item $\psi_{\lambda,*}$, $\psi_{\lambda,e}$ and $\psi_{\lambda,\pm}$ are the edge states at energy $\lambda$ of $\Di_*$, $\Di_e$ and $\Di_\pm$; see \eqref{eq:psi-lambda}, \eqref{eq:psi-lambda-e}. 
\end{itemize}

\subsection*{Acknowledgement} The author thanks Guillaume Bal for interesting conversations, as well as the NSF for support through the grants DMS-2439949 and DMS-2054589. The author acknowledges help from AI agent Claude through various parts of the project, in particular regarding the current identity \eqref{eq:current-identity} and the numerical simulations. 

\section{The unperturbed Dirac operator}\label{sec:unperturbed}

\subsection{Dirac operators with domain walls}\label{subsec:operator}

We work on the Hilbert space $L^2\left(\R^2,\C^2\right)$. Function spaces over $\R^2$ with values in $\C^2$ are abbreviated by dropping the arguments: we use the notations $L^2$, $H^1$, $L^2_\loc$ and $H^1_\loc$. Spaces over $\R$, such as $L^2(\R,\C^2)$, are spelled out. We write $D_j \defeq -i\partial_{x_j}$ for $j=1,2$.

\begin{definition}\label{def:domain-wall}
A horizontal domain wall is a function $\kappa_* : \R^2 \rightarrow \R$ of the form $\kappa_*(x) = k_*(x_2)$, where $k_* \in L^\infty(\R,\R)$ is such that
\begin{equation}\label{eq:domain-wall}
    k_*(x_2) = 1 \ \text{ for } x_2 \le -1, \qquad k_*(x_2) = -1 \ \text{ for } x_2 \ge 1.
\end{equation}
\end{definition}

\begin{definition}\label{def:dirac}
Given a horizontal domain wall $\kappa_*$, the unperturbed Dirac operator is
\begin{equation}\label{eq:dirac}
    \Di_* \defeq D_1 \sigma_1 + D_2 \sigma_2 + \kappa_*(x)\sigma_3,
\end{equation}
acting on $L^2$ with domain $H^1$.
\end{definition}

In \eqref{eq:dirac}, $\sigma_1, \sigma_2$ and $\sigma_3$ are the three Pauli matrices, given by:
\begin{equation}\label{eq:pauli}
    \sigma_1 \defeq \begin{bmatrix} 0 & 1 \\ 1 & 0 \end{bmatrix}, \qquad
    \sigma_2 \defeq \begin{bmatrix} 0 & -i \\ i & 0 \end{bmatrix}, \qquad
    \sigma_3 \defeq \begin{bmatrix} 1 & 0 \\ 0 & -1 \end{bmatrix}.
\end{equation}

The operator $D_1\sigma_1 + D_2\sigma_2$ is self-adjoint on $L^2$, with domain $H^1$. Moreover $\kappa_*\sigma_3$ is a bounded Hermitian multiplication operator; hence $\Di_*$ is self-adjoint on $L^2$, with domain $H^1$.

\subsection{The transverse Dirac operator} Since $\kappa_*$ depends on $x_2$ alone, $\Di_*$ commutes with translations in $x_1$. This suggests to introduce the transverse operator
\begin{equation}\label{eq:transverse}
    \Di_\perp \defeq D_2 \sigma_2 + k_*(x_2)\sigma_3
\end{equation}
on $L^2(\R,\C^2)$, with domain $H^1(\R,\C^2)$; it is self-adjoint on $L^2(\R,\C^2)$.

Let us recall a few spectral results from \cite[Appendix C]{DFW20} about $\Di_\perp$:
\begin{itemize}
    \item The absolutely continuous spectrum of $\Di_\perp$ is $\R \setminus (-1,1)$. 
    \item The spectrum of $\Di_\perp$ in $(-1,1)$ is made of an odd number of simple eigenvalues, symmetric about $0$. 
    \item $0$ is an eigenvalue of $\Di_\perp$.
\end{itemize}
It turns out that the zero mode of $\Di_\perp$ is explicit:

\begin{lemma}\label{lem:zero-mode}
The kernel of $\Di_\perp$ on $L^2(\R,\C^2)$ is spanned by
\begin{equation}\label{eq:phi0}
    \phi_*(x_2) \defeq c_* \begin{bmatrix} 1 \\ 1 \end{bmatrix} e^{K_*(x_2)}, \qquad K_*(x_2) \defeq \int_0^{x_2} k_*(s)\,ds,
\end{equation}
where $c_* > 0$ is the constant normalizing $\phi_*$ in $L^2(\R,\C^2)$. Moreover, $\sigma_1 \phi_* = \phi_*$ and $M_* \defeq \norm{e^{|x_2|}\phi_*}_\infty < +\infty$, so that $|\phi_*(x_2)| \le M_* e^{-|x_2|}$ for all $x_2 \in \R$.
\end{lemma}

The proof is a straightforward computation; see \cite[Appendix C]{DFW20}. The function
\begin{equation}\label{eq:psi-lambda}
    \psi_{\lambda,*}(x) \defeq e^{i\lambda x_1}\phi_*(x_2)
\end{equation}
satisfies $\Di_* \psi_{\lambda,*} = \lambda\,\psi_{\lambda,*}$. It is a plane wave in $x_1$, confined to $x_2$ near $0$, propagating rightwards for the Dirac equation $\left(D_t + \Di_*\right)\psi = 0$. Superposing it yields localized solutions of the form
\begin{equation}\label{eq:edge-packet}
    \psi(t,x) = a(x_1-t)\,\phi_*(x_2).
\end{equation}
In contrast, the other eigenvectors of $\Di_\perp$ do not seed generalized eigenmodes of $\Di_*$ that behave like plane waves in $x_1$. This is what makes \eqref{eq:psi-lambda} special.

\subsection{The edge decomposition}

\begin{definition}\label{def:edge-projection}
The edge projection $P_*$ is the orthogonal projection of $L^2$ defined by
\begin{equation}\label{eq:edge-projection}
    (P_*f)(x) \defeq a_*(x_1)\,\phi_*(x_2), \qquad
    a_*(x_1) \defeq \int_\R \ove{\phi_*(x_2)}^{\top}\, f(x_1,x_2)\, dx_2 .
\end{equation}
\end{definition}

We write $P_*^\perp \defeq \Id - P_*$ for the complementary projection, and
\begin{equation}\label{eq:HH}
    \HH_* \defeq \ran(P_*), \qquad \HH_*^\perp \defeq \ran\left(P_*^\perp\right),
\end{equation}
so that $L^2 = \HH_* \oplus \HH_*^\perp$. Accordingly, every $f \in L^2$ decomposes as
\begin{equation}\label{eq:decomposition}
    f(x) = a_*(x_1)\,\phi_*(x_2) + g(x), \qquad g \defeq P_*^\perp f .
\end{equation}
We recall that
\begin{equation}\label{eq:E-star}
    E_* \defeq \min\Big(\spec\left(\Di_\perp\right)\cap(0,+\infty)\Big).
\end{equation}

\begin{lemma}\label{lem:decomposition}
The projection $P_*$ commutes with $\Di_*$, and:
\begin{enumerate}[label=(\alph*)]
    \item on  $\HH_*$, the operator $\Di_*$ acts as $D_1$:
    \begin{equation}
        \Di_*\left(\beta(x_1)\phi_*(x_2)\right) = (D_1\beta)(x_1)\,\phi_*(x_2).
    \end{equation}
    \item the restriction of $\Di_*$ to $\HH_*^\perp$ has spectrum contained in $\Omega_*^c = \R\setminus\left(-E_*,E_*\right)$.
\end{enumerate}
\end{lemma}

\begin{proof} A direct computation based on \Cref{lem:zero-mode} shows (a).

For (b), recall that $\sigma_1$ anticommutes with $\sigma_2$ and with $\sigma_3$. Hence
\begin{equation}\label{eq:chiral}
    \sigma_1 \Di_\perp = -\Di_\perp \sigma_1 .
\end{equation}
Expanding the square and cancelling the cross terms,
\begin{equation}\label{eq:fiber-square}
    \Di_*^2 = D_1^2 + D_1\left(\sigma_1 \Di_\perp + \Di_\perp \sigma_1\right) + \Di_\perp^2 = D_1^2 + \Di_\perp^2 \geq \Di_\perp^2.
\end{equation}
By \Cref{lem:zero-mode}, the kernel of $\Di_\perp$ on $L^2(\R,\C^2)$ is $\C\phi_*$; and by \eqref{eq:E-star}, the spectrum of $\Di_\perp$ meets $\left(-E_*,E_*\right)$ only at $0$. Hence $\Di_\perp^2 \geq E_*^2$ on $\left(\C\phi_*\right)^\perp$. By arguing fiberwise, we deduce that  $\Di_\perp^2 \ge E_*^2$ as an operator on $\HH_*^\perp$. Therefore, $\Di_*^2 \geq E_*^2$ on $\HH_*^\perp$. Assertion (b) follows.
\end{proof}

\subsection{A current identity} We shall use below the quantity
\begin{equation}\label{eq:current}
    J(v,w) \defeq \ove{v}^\top \sigma w
    = \begin{bmatrix} \ove{v}^\top\sigma_1 w \\[4pt] \ove{v}^\top\sigma_2 w \end{bmatrix},
    \qquad v,w \in H^1_\loc.
\end{equation}
If $v,w \in H^1_\loc$ then $J(v,w) \in W^{1,1}_\loc$, and its trace on each circle $\partial B_R(0)$ is in $L^1$. In particular, the divergence theorem applies on balls.

We will need the following result:

\begin{lemma}\label{lem:edge-flux}
Let $\lambda \in \R$, let $\chi \in C^\infty(\R)$ such that for some $T > 0$, $\chi(t) = 0$ when $t \le -T$ and $\chi(t) = 1$ when $t \ge T$, and let $\zeta = \left(\zeta_1,\zeta_2\right) \in \R^2$ be a unit vector with $\zeta_1 > 0$. Set
\begin{equation}\label{eq:edge-flux-w}
    w_*(x) \defeq \chi\left(x\cdot\zeta\right) e^{i\lambda x_1}\, \phi_*(x_2) .
\end{equation}
For $R > 0$, we have:
\begin{equation}\label{eq:edge-flux}
    \int_{\partial B_R(0)} J(w_*,w_*)\cdot n \ d\ell = 1+O\left(e^{-\zeta_1 R/2}\right).
\end{equation}
\end{lemma}

\begin{proof} Because $\sigma_1\phi_* = \phi_*$ and $\ove{\phi_*}^\top\sigma_2\phi_* = 0$, we have
\begin{equation}\label{eq:edge-flux-J}
    J(w_*,w_*)(x) = \chi\left(x\cdot\zeta\right)^2 \begin{bmatrix} \ove{\phi_*(x_2)}^\top\sigma_1\phi_*(x_2) \\[4pt] \ove{\phi_*(x_2)}^\top\sigma_2\phi_*(x_2) \end{bmatrix}
    = \chi\left(x\cdot\zeta\right)^2 \left|\phi_*(x_2)\right|^2 e_1 .
\end{equation}

Because \eqref{eq:edge-flux-J} is carried by $e_1$, we have:
\begin{equation}\label{eq:edge-flux-divergence}
    \operatorname{div}J(w_*,w_*)(x) = \zeta_1\left(\chi^2\right)'\left(x\cdot\zeta\right)\left|\phi_*(x_2)\right|^2.
\end{equation}
The divergence theorem on $B_R(0)$ therefore gives
\begin{equation}\label{eq:edge-flux-div}
    \int_{\partial B_R(0)} J(w_*,w_*)\cdot n \ d\ell \; = \; \int_{B_R(0)} \zeta_1\left(\chi^2\right)'\left(x\cdot\zeta\right)\left|\phi_*(x_2)\right|^2 dx .
\end{equation}
For fixed $x_2$, the substitution $t = x\cdot\zeta$ has $dt = \zeta_1\,dx_1$; since $\chi^2$ runs from $0$ to $1$ and $\norm{\phi_*}_{L^2(\R)} = 1$, the same integrand over the whole plane gives
\begin{equation}\label{eq:edge-flux-full}
    \int_{\R^2} \zeta_1\left(\chi^2\right)'\left(x\cdot\zeta\right)\left|\phi_*(x_2)\right|^2 dx
    = \int_\R \left(\chi^2\right)'(t)\,dt \cdot \int_\R \left|\phi_*(x_2)\right|^2 dx_2 = 1 .
\end{equation}
It remains to discard the part of \eqref{eq:edge-flux-div} lying outside $B_R(0)$. Without loss of generality, assume $R \geq 2T/\zeta_1$. The function $\left(\chi^2\right)'$ is supported in $[-T,T]$, so the integral
\begin{equation}
    \int_{|x| > R} \zeta_1\left(\chi^2\right)'\left(x\cdot\zeta\right)\left|\phi_*(x_2)\right|^2 dx
\end{equation}
takes place on $\left|\zeta_1 x_1+\zeta_2 x_2\right| \le T$, where $\zeta_1\left|x_1\right| \le T + \left|x_2\right|$ and hence $\zeta_1 R \le \zeta_1|x| \le T+2|x_2|$; together with $R \ge 2T/\zeta_1$ this forces $|x_2| \ge \zeta_1 R/4$. Since $\phi_*$ decays exponentially in $x_2$, the same substitution gives
\begin{equation}
    \int_{|x| \ge R} \zeta_1\left|\left(\chi^2\right)'\left(x\cdot\zeta\right)\right| \left|\phi_*(x_2)\right|^2 dx \leq \int_\R \big| (\chi^2)'(t)\big| dt \cdot \int_{|x_2| \geq \zeta_1 R/4} \left|\phi_*(x_2)\right|^2 dx_2 = O\left(e^{-\zeta_1 R/2}\right).
\end{equation}
This proves \eqref{eq:edge-flux}.
\end{proof}

\section{The unperturbed resolvent}\label{sec:resolvent}

From now on we fix a horizontal domain wall $\kappa_*$.
The operator $\Di_*$ is self-adjoint, so its spectrum is real and for $\Im\lambda > 0$, its resolvent $(\Di_* - \lambda)^{-1}$ is a bounded operator on $L^2$. We prove here that $(\Di_* - \lambda)^{-1}$ continues analytically beyond $(-E_*,E_*)$ to a family of operators $R_*(\lambda)$, from $L^2_\comp$ to $H^1_\loc$. We also prove weighted estimates and spatial asymptotics formulas for $R_*(\lambda)$.

\begin{definition}\label{def:E-gamma}
For $\gamma\in(0,1)$, let $\EE_\gamma$ be the space
\begin{equation}\label{eq:E-gamma}
    \EE_\gamma \defeq \left\{ u \in L^2 \ : \ e^{\gamma|x|}u \in L^2\right\},
\end{equation}
normed by $\norm{u}_{\EE_\gamma} \defeq \norm{e^{\gamma|x|}u}_{L^2}$.
\end{definition}

The family of operators $(\Di_* - \lambda)^{-1}$, defined for $\lambda \in \C^+$, maps $L^2$ to $H^1$; in particular, it maps $L^2_\comp$ to $H^1_\loc$. The main theorem of this section shows that $(\Di_* - \lambda)^{-1} : L^2_\comp \rightarrow H^1_\loc$ can be analytically extended to
\begin{equation}
    \Omega_* \defeq \C \setminus \big( (-\infty,-E_*] \cup [E_*,+\infty) \big).
\end{equation}
With a view toward perturbing $\Di_*$, we also provide a spatial expansion of the continuation on weighted spaces when
\begin{equation}\label{eq:def-gamma}
    \lambda \in \tOmega_* \defeq \left\{ \lambda \in \C : |\lambda| < E_*, |\Im \lambda| < \gamma_*(\lambda) \right\}, \qquad \gamma_*(\lambda) \defeq \min\left( \dfrac{1}{2M_*^2}, \dfrac{E_*^2-|\lambda|^2}{20}\right).
\end{equation}
The set $\tOmega_*$ is a neighborhood of $\left(-E_*,E_*\right)$ in $\C$; see Figure \ref{fig:omega}.

\begin{figure}[ht]
\centering
\begin{minipage}[c]{0.54\textwidth}
    \centering
    \includegraphics[width=\linewidth]{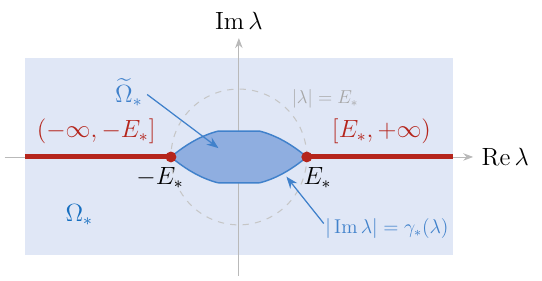}
\end{minipage}\hfill
\begin{minipage}[c]{0.42\textwidth}
\caption{The domain $\Omega_*$ of analytic continuation of $(\Di_*-\lambda)^{-1}$, and the subdomain $\tOmega_* \subset \Omega_*$ defined in \eqref{eq:def-gamma}, where the spatial asymptotics \eqref{eq:resolvent-asymptotics} hold.}
\label{fig:omega}
\end{minipage}
\end{figure}

\begin{theorem}\label{thm:resolvent-structure}   The family of operators $(\Di_* - \lambda)^{-1} : L^2_\comp \rightarrow H^1_\loc$, defined initially for $\lambda \in \C^+$, continues analytically to a family of operators
$R_*(\lambda) : L^2_\comp \rightarrow H^1_\loc$ for $\lambda \in \Omega_*$.

Moreover, for $\lambda \in \tOmega_*$ and $\gamma$ such that $|\Im\lambda| < \gamma \le \gamma_*(\lambda)$, the operator $R_*(\lambda)$ extends by density from $\EE_\gamma$ to $H^1_\loc$ and we have the spatial asymptotics
\begin{equations}\label{eq:resolvent-asymptotics}
    R_*(\lambda) f (x) =  i\int_{-\infty}^{x_1} e^{i\lambda(x_1-s)}a_*(s)\,ds \cdot \phi_*(x_2) + \EE_{\gamma/2}, \qquad f \in \EE_\gamma,
    \\
    P_*f(x) = a_*(x_1) \phi_*(x_2).
\end{equations}
\end{theorem}

We justify the convergence of the integral in \eqref{eq:resolvent-asymptotics} in \Cref{prop:edge-resolvent} below. The strategy in our proof of \Cref{thm:resolvent-structure} is to split $(\Di_*-\lambda)^{-1}$ using $P_*$:
\begin{equation}\label{eq:2a}
    (\Di_*-\lambda)^{-1} = (\Di_*-\lambda)^{-1}P_* + (\Di_*-\lambda)^{-1} P_*^\perp.
\end{equation}
We then extend each term in the right-hand-side analytically.

\subsection{The edge decomposition}\label{subsec:edge-decomposition}

We first study the analytic continuation of $(\Di_*-\lambda)^{-1} P_*$. For $\Im \lambda > 0$, this operator maps $L^2$ to $L^2$, hence $L^2_\comp$ to $L^2_\loc$.

\begin{proposition}\label{prop:edge-resolvent}
The family of operators $(\Di_* - \lambda)^{-1}P_* : L^2_\comp \rightarrow L^2_\loc$, defined initially for $\Im \lambda > 0$, continues analytically to a family of operators
$A_*(\lambda) : L^2_\comp \rightarrow L^2_\loc$ for $\lambda \in \Omega_*$.

Moreover, for $\lambda \in \tOmega_*$ and $\gamma$ such that $|\Im\lambda| < \gamma \le \gamma_*(\lambda)$, the operator $A_*(\lambda)$ extends by density from $\EE_\gamma$ to $L^2_\loc$, and we have the formula:
\begin{equation}\label{eq:edge-resolvent}
   A_*(\lambda)f(x) = i\int_{-\infty}^{x_1} e^{i\lambda(x_1-s)}a_*(s)\,ds \cdot \phi_*(x_2), \quad f\in \EE_\gamma, \quad  P_*f(x) = a_*(x_1) \phi_*(x_2),
\end{equation}
where the above integral converges absolutely.
\end{proposition}

\begin{proof}[Proof of \Cref{prop:edge-resolvent}] \textbf{1.} Before all, we justify that for $\lambda \in \tOmega_*$, $\gamma$ such that $|\Im\lambda| < \gamma \le \gamma_*(\lambda)$, $f \in \EE_\gamma$, and $a_*$ as in \eqref{eq:edge-projection}, the integral in \eqref{eq:edge-resolvent} converges absolutely. 
    
    We claim that $a_* \in e^{-\gamma|x_1|}L^2(\R)$. Indeed, using the formula \eqref{eq:edge-projection} for $a_*$, we have:
\begin{equation}
    e^{\gamma|x_1|}\left|a_*(x_1)\right| \; \le \; \int_\R \left|\phi_*(x_2)\right| e^{\gamma|x|}\left|f(x)\right| dx_2 \; \le \; \norm{\phi_*}_{L^2(\R)} \left(\int_\R e^{2\gamma|x|}\left|f(x)\right|^2 dx_2\right)^{1/2} .
\end{equation}
Since $\norm{\phi_*}_{L^2(\R)} = 1$, integrating in $x_1$ produces \eqref{eq:edge-data}:
\begin{equation}\label{eq:edge-data}
   \norm{e^{\gamma|x_1|}a_*}_{L^2(\R)}^2 = \int_\R e^{2\gamma|x_1|}\left|a_*(x_1)\right|^2 dx_1 \; \le \; \int_{\R^2} e^{2\gamma|x|}\left|f(x)\right|^2 dx = \|f\|_{\EE_\gamma}^2.
\end{equation}

We now note that: 
\begin{equations}\label{eq:edge-integral-bound}
    \int_{-\infty}^{x_1} \left| e^{-i\lambda s} a_*(s) \right| ds \leq \int_{-\infty}^{x_1} \left| e^{|\Im \lambda|\,|s|}a_*(s) \right| ds
    \\
     = \int_{-\infty}^{x_1} \left| e^{ (|\Im \lambda|-\gamma)|s|} \cdot e^{ \gamma|s|} a_*(s) \right| ds \leq \dfrac{1}{\sqrt{\gamma-|\Im \lambda|}} \norm{e^{\gamma|x_1|}a_*}_{L^2(\R)}\leq \dfrac{1}{\sqrt{\gamma-|\Im \lambda|}} \norm{f}_{\EE_\gamma},
\end{equations}
where in the last line we used the Cauchy--Schwarz inequality and $|\Im \lambda| < \gamma$. This proves that the integral in \eqref{eq:edge-resolvent} converges absolutely.

    \textbf{2.} We now prove the core of the proposition. First work with $\lambda$ with $\Im\lambda>0$ and $f \in L^2_\comp$. We look for $u(x) = \beta(x_1)\phi_*(x_2) \in \HH_*$ solving $(\Di_*-\lambda)u = P_*f = a_*\phi_*$. Note that $a_* \in L^2_\comp(\R)$: indeed, if $\supp(f) \subset B_R(0)$, for $|x_1| \geq R$ we have:
    \begin{equation}
        a_*(x_1) = \int_\R \ove{\phi_*(x_2)}^\top f(x_1,x_2) dx_2 = 0.
    \end{equation}
Because $a_* \in L^2(\R)$ -- see \eqref{eq:edge-data} -- we deduce $a_* \in L^2_\comp$. By part (a) of \Cref{lem:decomposition}, solving $(\Di_*-\lambda)u = a_*\phi_*$ in $\HH_*$ reduces to the scalar ODE
\begin{equation}\label{eq:transport}
    (D_1 - \lambda)\beta = a_*.
\end{equation}

\textbf{3.} Because $f \in L^2_\comp \subset L^2$, $a_* \in L^2$. Moreover $D_1$ is selfadjoint on $L^2(\R)$ and $\Im \lambda > 0$. Hence $\beta = (D_1-\lambda)^{-1} a_*$ is the unique solution to the equation \eqref{eq:transport} in $L^2(\R)$. It is given by the integral formula
\begin{equation}\label{eq:edge-term}
    \beta(x_1) = (D_1-\lambda)^{-1} a_* = i\,e^{i\lambda x_1}\int_{-\infty}^{x_1} e^{-i\lambda s}a_*(s)\,ds
    = i\int_{-\infty}^{x_1} e^{i\lambda(x_1-s)}a_*(s)\,ds .
\end{equation}
Indeed, a straightforward calculation shows that \eqref{eq:edge-term} solves \eqref{eq:transport}. Moreover, \eqref{eq:edge-term} is square-integrable, because it is the convolution of $a_* \in e^{-\gamma|x_1|}L^2(\R) \subset L^2(\R)$ against the kernel
\begin{equation}\label{eq:transport-kernel}
    k_\lambda(t) \defeq e^{i\lambda t}\,\1_{t \ge 0}, \qquad \text{which has} \qquad \norm{k_\lambda}_{L^1(\R)} = \int_0^{+\infty} e^{-t\Im\lambda}\,dt = \frac{1}{\Im\lambda} < +\infty .
\end{equation}
Young's convolution inequality therefore gives $\norm{\beta}_{L^2(\R)} \le (\Im\lambda)^{-1}\norm{a_*}_{L^2(\R)}$, in particular $\beta \in L^2(\R)$. We deduce that for $\Im \lambda >0$ and $f \in L^2_\comp$,
\begin{equation}\label{eq:edge-resolvent-comp}
   A_*(\lambda)f(x) = i\int_{-\infty}^{x_1} e^{i\lambda(x_1-s)}a_*(s)\,ds \cdot \phi_*(x_2) .
\end{equation}

\textbf{4.} For $f \in L^2_\comp$ and $\lambda \in \Omega_*$, $a_* \in L^2_\comp(\R)$, and the integral \eqref{eq:edge-resolvent-comp} defines an element of $L^2_\loc$: indeed,  for any $L$:
\begin{equations}
    \int_{-L}^L \left| ie^{i\lambda x_1} \int_{-\infty}^{x_1} e^{-i\lambda s}a_*(s)\,ds  \right|^2 dx_1  \leq \; 2L \cdot 2R\, e^{2\left|\Im\lambda\right|(L+R)} \| a_* \|_{L^2(\R)}^2 \; < \; \infty,
\end{equations}
where $R$ is such that $\supp(a_*) \subset [-R,R]$, and where we used $\left|x_1-s\right| \le L+R$ on the domain of integration; the right-hand side is finite for every $\lambda \in \C$. Moreover, the formula \eqref{eq:edge-resolvent-comp} depends analytically on $\lambda$ as can be seen from an application of Morera's theorem. Because $\phi_* \in L^2(\R)$ is normalized, this extension satisfies
\begin{equation}
\norm{A_*(\lambda) f}^2_{L^2((-L,L) \times \R)} = \int_{-L}^L \left| \int_{-\infty}^{x_1} e^{i\lambda(x_1-s)}a_*(s)\,ds  \right|^2 dx_1 \; < \; \infty .
\end{equation}
Therefore it defines an analytic family of operators $A_*(\lambda) : L^2_\comp \rightarrow L^2_\loc$ for $\lambda \in \Omega_*$, that extends $(\Di_* - \lambda)^{-1}P_* : L^2_\comp \rightarrow L^2_\loc$ past $\left(-E_*,E_*\right)$.

\textbf{5.} Assume now $\lambda \in \tOmega_*$ and $\gamma$ such that $|\Im\lambda| < \gamma \le \gamma_*(\lambda)$. We justify that $A_*(\lambda)$ extends by density as an operator from $\EE_\gamma$ to $L^2_\loc$. Let $f \in L^2_\comp$; in particular, $f \in \EE_\gamma$ and $e^{\gamma|x_1|}a_* \in L^2(\R)$. By \eqref{eq:edge-integral-bound}, we have, for any $L$:
\begin{align}\label{eq:edge-density-bound}
    \int_{-L}^L \left| ie^{i\lambda x_1} \int_{-\infty}^{x_1} e^{-i\lambda s}a_*(s)\,ds  \right|^2 dx_1 & \leq \int_{-L}^L e^{2|\Im \lambda| |x_1|} \cdot \dfrac{1}{\gamma-|\Im \lambda|} \norm{e^{\gamma|x_1|}a_*}_{L^2(\R)}^2 dx_1
    \\
&     \leq \dfrac{2L e^{2L}}{\gamma-|\Im \lambda|} \| f \|_{\EE_\gamma}^2.
\end{align}
Because of \eqref{eq:edge-density-bound}, we conclude that $A_*(\lambda)$ maps $\EE_\gamma$ to $L^2_\loc$ indeed. This completes the proof.
\end{proof}

\subsection{The resolvent on \texorpdfstring{$\HH_*^\perp$}{H-perp}}\label{subsec:resolvent-complement}

By part (b) of \Cref{lem:decomposition}, every $\lambda \in \Omega_*$ lies in the resolvent set of $\Di_*|_{\HH_*^\perp}$, so the operator
\begin{equation}\label{eq:A-lambda}
    B_*(\lambda) \defeq (\Di_*-\lambda)^{-1} P_*^\perp : L^2 \rightarrow L^2
\end{equation}
has an analytic continuation for $\lambda \in \Omega_*$. In particular, $B_*(\lambda)$ maps $L^2_\comp$ to $L^2_\loc$.
 We prove here that it satisfies weighted estimates:

\begin{proposition}\label{prop:agmon} Let $\zeta \in \R^2$ be a unit vector, $|\lambda| < E_*$, and $0 < \epsi \leq 2\gamma_*(\lambda)$, with $\gamma_*(\lambda)$ given by \eqref{eq:def-gamma}.
For every $f \in \HH_*^\perp$:
\begin{equation}\label{eq:agmon-estimate}
    e^{\epsi |x\cdot \zeta|} f \in L^2 \quad \Rightarrow \quad
    e^{\epsi|x \cdot \zeta|} B_*(\lambda) f \in L^2, \quad
    \big\| e^{\epsi|x \cdot \zeta|} B_*(\lambda) f \big\| \leq 4\epsi^{-1} \big\|e^{\epsi |x \cdot \zeta|} f\big\|.
\end{equation}
\end{proposition}

In particular, by the closed graph theorem, the operator $B_*(\lambda)$ is bounded on $e^{-\epsi|x\cdot\zeta|}L^2$, with norm less than $4\epsi^{-1}$.

\begin{proof} The proof relies on Agmon estimates, which are themselves an elaborate version of the identity
    \begin{equations}\label{eq:1e}
        \lr{(S+\lambda-iT)v, (S-\lambda+iT)v}
        \\
        = \| S v\|^2-|\lambda|^2 \| v\|^2 - \| T v\|^2 + 2i \Re \lr{S v,Tv} - 2i \Im \lambda \lr{v,Sv} + 2 \Im \lambda\lr{v,Tv},
    \end{equations}
valid for selfadjoint operators $S$ and $T$ and $\lambda \in \C$. Taking the real part of \eqref{eq:1e} gives:
\begin{equation}\label{eq:1k}
        \Re \lr{(S+\lambda-iT)v, (S-\lambda+iT)v} 
        = \| S v\|^2-|\lambda|^2 \| v\|^2 - \| T v\|^2 + 2 \Im \lambda\lr{v,Tv}.
    \end{equation}

\textbf{1.} Let $|\lambda|<E_*$ and $0<\epsi \leq 2\gamma_*(\lambda)$ with $\gamma_*(\lambda)$ given by \eqref{eq:def-gamma}; note that $\epsi \in (0,1)$. Fix a unit vector $\zeta \in \R^2$, $N \in \N$, and define:
\begin{equations}\label{eq:1i}
    \varphi_k(x_k) = \begin{cases}
        -N & \quad \text{for } \epsi \zeta_k x_k \leq -N
        \\
        \epsi \zeta_k x_k & \quad \text{for } |\epsi \zeta_k x_k| \leq N
        \\
        N & \quad \text{for } \epsi \zeta_k x_k \geq N
    \end{cases}, \qquad
    \varphi(x) \defeq \varphi_1(x_1) + \varphi_2(x_2).
\end{equations}
The function $\varphi$ is bounded and Lipschitz (in particular it is differentiable almost everywhere) and $\| \nabla \varphi \|_\infty \leq \epsi$. Moreover, $|\varphi(x)| \leq \epsi |\zeta \cdot x|$. Indeed, write
\begin{equation}
    \varphi(x) = \pi_N(\epsi \zeta_1 x_1) + \pi_N(\epsi \zeta_2 x_2), \qquad \pi_N(s) = \max(-N,\min(N,s)).
\end{equation}
The map $\pi_N$ preserves sign. Therefore, if $\zeta_1 x_1$ and $\zeta_2 x_2$ have the same sign, so do $\pi_N(\epsi \zeta_1 x_1)$ and $\pi_N(\epsi \zeta_2 x_2)$. Hence, because $\pi_N$ is $1$-Lipschitz, 
\begin{equation}
    |\varphi(x)| = |\pi_N(\epsi \zeta_1 x_1)| + |\pi_N(\epsi \zeta_2 x_2)| \leq \epsi |\zeta_1 x_1| + \epsi |\zeta_2 x_2| = \epsi |x \cdot \zeta|.
\end{equation}
If on the other hand, $\zeta_1 x_1$ and $\zeta_2 x_2$ have different signs, so do $\pi_N(\epsi \zeta_1 x_1)$ and $\pi_N(\epsi \zeta_2 x_2)$. Because $|\pi_N(s)| = \min(N,|s|)$, we have $||\pi_N(s)|-|\pi_N(t)|| \leq ||s|-|t||$, so
\begin{equation}
    \big|\varphi(x) \big| \; = \; \Big| |\pi_N(\epsi \zeta_1 x_1)| - |\pi_N(\epsi \zeta_2 x_2)| \Big| \leq \epsi \big| |\zeta_1 x_1|- |\zeta_2 x_2| \big| = \epsi |\zeta \cdot x|.
\end{equation}
This proves that $|\varphi(x)| \leq \epsi |\zeta \cdot x|$ indeed.

Let $f \in \HH_*^\perp$ such that $e^{\epsi |x\cdot \zeta|} f \in L^2$, and set $u = B_*(\lambda) f \in \HH_*^\perp$.
The function $v=e^\varphi u$ satisfies $D_k v = e^\varphi D_ku -i \partial_k\varphi \cdot v$, so
\begin{equation}\label{eq:1q}
    (\Di_*-\lambda + i (\sigma \cdot \nabla)\varphi) v = e^\varphi (\Di_*-\lambda) u = e^\varphi f.
\end{equation}
Pair both sides with $(\Di_*+\lambda - i (\sigma \cdot \nabla)\varphi) v = e^\varphi f+ 2\lambda v - 2i (\sigma \cdot \nabla)\varphi \cdot v$ to obtain
\begin{equation}\label{eq:1t}
    \lr{(\Di_*+\lambda - i (\sigma \cdot \nabla)\varphi) v, (\Di_*-\lambda + i (\sigma \cdot \nabla)\varphi) v} =  \big\| e^\varphi f \big\|^2 + 2\lr{(\lambda - i(\sigma \cdot \nabla)\varphi) v, e^\varphi f}.
\end{equation}
Take the real part of the left-hand-side and expand it according to  \eqref{eq:1k} to obtain:
\begin{equations}\label{eq:1l}
    \big\| \Di_* v \big\|^2 - |\lambda|^2 \big\| v \big\|^2 - \big\| (\sigma \cdot \nabla)\varphi\, v \big\|^2 + 2 \Im \lambda \lr{(\sigma \cdot \nabla)\varphi\, v, v}
    \\
    =  \big\| e^\varphi f \big\|^2 + 2\Re \lr{(\lambda - i(\sigma \cdot \nabla)\varphi) v, e^\varphi f}.
\end{equations}

Recall that $|\lambda| < E_*$. We have the following inequalities:
\begin{align}
    \big| 2 \Im \lambda \lr{(\sigma \cdot \nabla)\varphi\, v, v} \big|
    & \leq \epsi^{-1} \big\| (\sigma \cdot \nabla)\varphi\, v \big\|^2 + \epsi \big\| v \big\|^2
    \leq \Big(\epsi^{-1} \big\| \nabla \varphi \big\|_\infty^2 + \epsi\Big) \big\| v \big\|^2 ,
    \\
    \big| 2 \Re \lr{\lambda v, e^\varphi f} \big|
    & \leq \epsi \big\| v \big\|^2 + \epsi^{-1} \big\| e^\varphi f \big\|^2 ,
    \\
    \big| 2 \Re \lr{-i(\sigma \cdot \nabla)\varphi\, v, e^\varphi f} \big|
    & \leq \epsi \big\| (\sigma \cdot \nabla)\varphi\, v \big\|^2 + \epsi^{-1} \big\| e^\varphi f \big\|^2
    \leq \epsi \big\| \nabla \varphi \big\|_\infty^2 \big\| v \big\|^2 + \epsi^{-1} \big\| e^\varphi f \big\|^2 .
\end{align}
Plug them in \eqref{eq:1l} to obtain:
\begin{equations}\label{eq:1m}
    \big\| \Di_* v \big\|^2 - \Big(|\lambda|^2+2\epsi + (1+\epsi^{-1}+\epsi) \big\| \nabla \varphi \big\|_\infty^2 \Big) \big\| v \big\|^2
    \leq (1+2\epsi^{-1}) \big\| e^\varphi f \big\|^2.
\end{equations}
Because $\| \nabla \varphi \|_\infty \leq \epsi$ and $\epsi \in (0,1)$, we have $(1+\epsi^{-1}+\epsi) \| \nabla \varphi \|_\infty^2 \leq 3\epsi$ and $1+2\epsi^{-1} \leq 3\epsi^{-1}$. This implies
\begin{equations}\label{eq:1h}
    \big\| \Di_* v \big\|^2 - \big(|\lambda|^2+5\epsi \big) \big\| v \big\|^2
    \leq 3\epsi^{-1} \big\| e^\varphi f \big\|^2.
\end{equations}

\textbf{2.} From \Cref{lem:decomposition}, we get:
\begin{equation}\label{eq:1u}
    \big\| \Di_* v \big\|^2 \geq \big\| \Di_* P_*^\perp v \big\|^2 \geq E_*^2\big\| P_*^\perp v \big\|^2 = E_*^2\Big(\big\| v \big\|^2 - \big\| P_*v \big\|^2\Big) .
\end{equation}
Plugging this in \eqref{eq:1h} produces
\begin{equation}\label{eq:1n}
    \big(E_*^2-|\lambda|^2-5\epsi \big) \big\| v \big\|^2 - E_*^2\big\| P_*v \big\|^2 \leq \big\| \Di_* v \big\|^2 - \big(|\lambda|^2+5\epsi \big) \big\| v \big\|^2 \leq  3\epsi^{-1} \big\| e^\varphi f \big\|^2.
\end{equation}

\textbf{3.} We bound $\| P_*v \|$. Recall that  $\varphi = \varphi_1+\varphi_2$ was defined in \eqref{eq:1i}.
Set $w = e^{\varphi_2} u$ and note that $0 = P_*u = P_* e^{-\varphi_2} w$. Moreover, $P_*$ commutes with $e^{\varphi_1}$. Therefore,
\begin{equation}
    P_*v = P_* e^{\varphi_1+\varphi_2} u = e^{\varphi_1} P_* w = e^{\varphi_1} P_*(1-e^{-\varphi_2}) w.
 \end{equation}
It follows that:
\begin{equation}\label{eq:1p}
    \big\| P_*v \big\|^2 = \int_{\R^2} e^{2\varphi_1(x_1)} \left|\phi_*(x_2)\right|^2
    \left| \int_\R \ove{\phi_*(y_2)}^\top \left(1-e^{-\varphi_2(y_2)}\right) w(x_1,y_2)\,dy_2 \right|^2 dx .
\end{equation}
Below we use the notation $\Phi(y_2) = \ove{\phi_*(y_2)}^\top (1-e^{-\varphi_2(y_2)})$ and we estimate the integral \eqref{eq:1p}. First realize the integral in $x_2$, using that $\phi_*$ is normalized. Then use the Cauchy--Schwarz inequality on the integral in $y_2$ to obtain:
\begin{align}
    \big\| P_*v \big\|^2 & = \int_\R e^{2\varphi_1(x_1)}
    \left| \int_\R \Phi(y_2)\, w(x_1,y_2)\,dy_2 \right|^2 dx_1
    \\
    & \leq \big\| \Phi \big\|^2_{L^2(\R)} \int_\R e^{2\varphi_1(x_1)}
    \big\| w(x_1,\cdot) \big\|^2_{L^2(\R)} dx_1
     \; = \; \big\| \Phi \big\|^2_{L^2(\R)} \big\| e^{\varphi_1} w \big\|^2 \; = \; \big\| \Phi \big\|^2_{L^2(\R)} \big\| v \big\|^2 .
\end{align}

We now estimate $\| \Phi \|_{L^2(\R)}$. The bounds $\left|1-e^{-s}\right| \leq |s|e^{|s|}$ and $|\varphi_2(x_2)| \leq \epsi |x_2|$ yield 
\begin{equation}\label{eq:Phi-pointwise}
    \big|1-e^{-\varphi_2(x_2)}\big| \; \leq \; \epsi\, |x_2|\, e^{\epsi|x_2|} .
\end{equation}
Recall from \Cref{lem:zero-mode} that $\left|\phi_*(x_2)\right| \leq M_* e^{-|x_2|}$. Since $\epsi \leq 1/10$, we deduce that:
\begin{equation}
    \big\| \Phi \big\|_{L^2(\R)}^2 \; \leq \; M_*^2 \epsi^2 \int_\R x_2^2\, e^{-2(1-\epsi)|x_2|}\,dx_2 \; \leq \; M_*^2 \epsi^2 \int_\R x_2^2\, e^{-|x_2|}\,dx_2 \; = \; 4M_*^2 \epsi^2 .
\end{equation}
It follows that $\| P_*v \| \leq 2M_*\epsi \| v\|$.

\textbf{4.} Plugging this bound in \eqref{eq:1n} produces:
\begin{equation}\label{eq:1o}
    \big(E_*^2-|\lambda|^2-5\epsi - 4M_*^2 E_*^2\epsi^2 \big) \big\| v \big\|^2 \leq  3\epsi^{-1} \big\| e^\varphi f \big\|^2.
\end{equation}
Because $\epsi \leq 2\gamma_*(\lambda)$ with $\gamma_*(\lambda)$ given by \eqref{eq:def-gamma}, we have $\epsi \leq M_*^{-2}$ and $\epsi \leq \tfrac{E_*^2-|\lambda|^2}{10}$. Then $4M_*^2E_*^2\epsi^2 \leq 4E_*^2\epsi \leq 4\epsi$, so that
\begin{equation}\label{eq:1r}
    E_*^2-|\lambda|^2-5\epsi - 4M_*^2 E_*^2\epsi^2 \geq E_*^2-|\lambda|^2 - 9\epsi \geq \dfrac{E_*^2-|\lambda|^2}{10} \geq \epsi.
\end{equation}
Plugging \eqref{eq:1r} in \eqref{eq:1o} (and recalling that $|\lambda| < E_* \leq 1$), we obtain
\begin{equation}
    \big\| v \big\| \leq 2\epsi^{-1} \big\| e^\varphi f \big\|.
\end{equation}

\textbf{5.} Recall that $\varphi$ depends on $N$. We now emphasize this dependence by writing $\varphi = \varphi^{(N)}$ below. As $N \rightarrow \infty$, $\varphi^{(N)}$ converges to $\epsi \zeta\cdot x$. Therefore, by Fatou's lemma:
\begin{equations}
   \big\| e^{\epsi x \cdot \zeta} u \big\|^2 = \int_{\R^2} e^{2\epsi\, x \cdot \zeta} \left|u\right|^2 dx = \int_{\R^2} \liminf_{N \rightarrow \infty} e^{2\varphi^{(N)}} \left|u\right|^2 dx
   \\
   \leq \liminf_{N \rightarrow \infty} \int_{\R^2} e^{2\varphi^{(N)}} \left|u\right|^2 dx = \liminf_{N \rightarrow \infty} \big\| e^{\varphi^{(N)}} u \big\|^2 \leq 4\epsi^{-2} \cdot \liminf_{N \rightarrow \infty} \big\| e^{\varphi^{(N)}} f \big\|^2 \leq 4\epsi^{-2} \cdot\big\| e^{\epsi|x \cdot \zeta|} f \big\|^2,
\end{equations}
where we used $|\varphi (x)| \leq \epsi|x \cdot \zeta|$ in the last line.

In particular, $e^{\epsi x \cdot \zeta} u$ is in $L^2$. Changing $\zeta$ to $-\zeta$ produces $e^{-\epsi x \cdot \zeta} u \in L^2$ with the same bound. We can then conclude using the inequality $e^{\epsi |x \cdot \zeta|} |u| \leq e^{\epsi x \cdot \zeta} |u| +  e^{-\epsi x \cdot \zeta} |u|$.
It implies, by the triangle inequality on $L^2$:
\begin{equation}
    \big\| e^{\epsi |x \cdot \zeta|} u \big\| \leq \big\| e^{\epsi\, x \cdot \zeta} u \big\|  + \big\| e^{-\epsi x \cdot \zeta} u \big\| \leq 4\epsi^{-1} \cdot \big\| e^{\epsi |x \cdot \zeta|} f \big\|.
\end{equation}
This proves \eqref{eq:agmon-estimate} when $f \in \HH_*^\perp$.
\end{proof}

\begin{lemma}\label{lem:projection-weight} Let $\zeta \in \R^2$ be a unit vector and $0 < \epsi < 1$.
    If $e^{\epsi |x \cdot \zeta|} f \in L^2$, then $e^{\epsi |x \cdot \zeta|} P_*f$ and $e^{\epsi |x \cdot \zeta|} P_*^\perp f$ are both in $L^2$.
\end{lemma}

By the closed graph theorem, \Cref{lem:projection-weight} is equivalent to saying that $P_*$ and $P_*^\perp$ are bounded operators on $e^{-\epsi |x \cdot \zeta|} L^2$.

\begin{proof} Set $\varphi(x) = \varphi_1(x_1)+\varphi_2(x_2)$ as in \eqref{eq:1i}. We have:
\begin{equations}
    \norm{ e^{\varphi} P_* f }^2 = \int_{\R^2} e^{2\varphi_1(x_1)} e^{2\varphi_2(x_2)} |\phi_*(x_2)|^2 \left|\int_\R \ove{\phi_*(y_2)}^\top f(x_1,y_2) dy_2 \right|^2 dx
    \\
    = \int_\R e^{2\varphi_1(x_1)} \left|\int_\R \ove{\phi_*(y_2)}^\top f(x_1,y_2) dy_2 \right|^2 dx_1 \cdot \int_\R e^{2\varphi_2(x_2)} |\phi_*(x_2)|^2 dx_2
    \\
    \leq \dfrac{M_*^2}{1-\epsi} \cdot \int_\R e^{2\varphi_1(x_1)} \left|\int_\R e^{-\varphi_2(y_2)}\ove{\phi_*(y_2)}^\top \cdot e^{\varphi_2(y_2)} f(x_1,y_2) dy_2 \right|^2 dx_1
    \\
    \leq \dfrac{M_*^4}{(1-\epsi)^2} \cdot \int_\R e^{2\varphi_1(x_1)} \int_\R e^{2\varphi_2(y_2)} \left|f(x_1,y_2)\right|^2 dy_2\, dx_1 = \dfrac{M_*^4}{(1-\epsi)^2} \| e^{\varphi} f\|^2.
\end{equations}
In the first line, we used the definition of $P_*$ and $\varphi(x) = \varphi_1(x_1)+\varphi_2(x_2)$; in the second, Fubini's theorem; in the third, the bound $|\phi_*(x_2)| \leq M_* e^{-|x_2|}$; in the fourth, we used the Cauchy--Schwarz inequality. Using an argument similar to Step 5 in the proof of \Cref{prop:agmon}, we conclude that
\begin{equation}\label{eq:1w}
    \norm{ e^{\epsi |x\cdot \zeta|} P_* f } \leq \dfrac{2M_*^2}{1-\epsi} \| e^{\epsi |x\cdot \zeta|} f\|.
\end{equation}
This proves $e^{\epsi |x \cdot \zeta|} P_* f \in L^2$. We obtain $e^{\epsi |x \cdot \zeta|} P_*^\perp f \in L^2$ by using $P_*^\perp f =f-P_*f$. This completes the proof.
\end{proof}

\begin{corollary}\label{prop:agmon-preview} Let $|\lambda| < E_*$, $0 < \gamma \leq \gamma_*(\lambda)$, with $\gamma_*(\lambda)$ given by \eqref{eq:def-gamma}. The operator $B_*(\lambda)$ is bounded from $\EE_\gamma$ to $\EE_{\gamma/2}$; more specifically:
    \begin{equations}\label{eq:agmon-radial}
        \norm{B_*(\lambda)f}_{\EE_{\gamma/2}} \leq 2^5 M_*^2 \gamma^{-1} \norm{f}_{\EE_\gamma}.
    \end{equations}
\end{corollary}

\begin{proof} Note the inequality:
    \begin{equation}
        e^{\gamma|x|} \le e^{2\gamma \max(|x_1|,|x_2|)} \leq e^{2\gamma|x_1|} + e^{2\gamma|x_2|}.
    \end{equation}
    Let $f \in \EE_\gamma$ and $k \in \{1,2\}$. Since $e^{\gamma|x_k|} \leq e^{\gamma|x|}$, we have $e^{\gamma|x_k|} f \in L^2$. By \eqref{eq:1w}, using $P_*^\perp = \Id-P_*$:
    \begin{equation}\label{eq:agmon-preview-projection}
        \norm{e^{\gamma|x_k|} P_*^\perp f} \leq \left(1 + \dfrac{2M_*^2}{1-\gamma}\right) \norm{e^{\gamma|x_k|} f} \leq 2^2 M_*^2 \norm{f}_{\EE_\gamma}.
    \end{equation}
    In the last inequality, we used $\gamma \leq \gamma_*(\lambda) \leq 1/20$ and $M_* \geq 1$ (which we get from $1 = \norm{\phi_*}_{L^2(\R)}^2 \leq M_*^2 \int_\R e^{-2|x_2|} dx_2 = M_*^2$). Moreover, $B_*(\lambda) f = B_*(\lambda) P_*^\perp f$ with $P_*^\perp f \in \HH_*^\perp$.
    Therefore, using \Cref{prop:agmon} applied to $P_*^\perp f \in \HH_*^\perp$:
    \begin{align}
        \norm{B_*(\lambda)f}_{\EE_{\gamma/2}}^2 & \leq \big\|e^{\gamma |x_1|} B_*(\lambda)f\big\|^2 + \big\|e^{\gamma |x_2|} B_*(\lambda)f\big\|^2
        \\
        & \leq 2^4 \gamma^{-2} \left( \big\|e^{\gamma |x_1|} P_*^\perp f\big\|^2 + \big\|e^{\gamma |x_2|} P_*^\perp f\big\|^2 \right) \leq 2^9 M_*^4 \gamma^{-2} \norm{f}_{\EE_\gamma}^2.
    \end{align}
    This completes the proof.
\end{proof}

We mention that $B_*(\lambda)$ is most likely bounded from $\EE_\gamma$ to itself, though we will not need this in this work.

\subsection{The full resolvent}\label{subsec:resolvent-structure}

It remains to recombine the two halves of the decomposition.

\begin{proof}[Proof of \Cref{thm:resolvent-structure}]
\textbf{1.} We note that for $\Im\lambda>0$:
\begin{equation}\label{eq:resolvent-split}
    (\Di_*-\lambda)^{-1} \; = \; (\Di_*-\lambda)^{-1}P_* \; + \; (\Di_*-\lambda)^{-1}P_*^\perp \; = \; A_*(\lambda)\ \; + \; B_*(\lambda),
\end{equation}
where each term maps $L^2$ to $L^2$, in particular $L^2_\comp$ to $L^2_\loc$.

By \Cref{prop:edge-resolvent} and the discussion preceding \Cref{prop:agmon-preview}, each term in \eqref{eq:resolvent-split} continues analytically to $\lambda \in \Omega_*$ as operators from $L^2_\comp$ to $L^2_\loc$. This implies that $(\Di_*-\lambda)^{-1}$ continues analytically to $\lambda \in \Omega_*$ as an operator $R_*(\lambda)$ from $L^2_\comp$ to $L^2_\loc$.

According to \Cref{prop:edge-resolvent} and \eqref{eq:agmon-radial}, for $\lambda \in \tOmega_*$ and $\gamma$ such that $|\Im\lambda| < \gamma \le \gamma_*(\lambda)$, both $A_*(\lambda)$ and $B_*(\lambda)$ extend by density as operators $\EE_\gamma \rightarrow L^2_\loc$. Therefore, so does $R_*(\lambda)$. Finally, we obtain the spatial asymptotics \eqref{eq:resolvent-asymptotics} for $R_*(\lambda)$ by simply  plugging the spatial asymptotics \eqref{eq:edge-resolvent} and \eqref{eq:agmon-radial} for $A_*(\lambda)$ and $B_*(\lambda)$.

It remains to upgrade $L^2_\loc$ to $H^1_\loc$ in both statements. In the sense of distributions, we have for every $f \in L^2_\comp$, respectively every $f \in \EE_\gamma$,
\begin{equation}\label{eq:resolvent-equation}
    \left(\Di_*-\lambda\right)R_*(\lambda)f = f.
\end{equation}
Indeed, this identity holds when $\Im \lambda > 0$ and extends analytically to all $\lambda \in \Omega_*$. In particular, $u = R_*(\lambda)f$ satisfies $\left(\Di_*-\lambda\right) u \in L^2_\loc$. Elliptic regularity implies that $u \in H^1_\loc$. By the closed graph theorem, the operator $R_*(\lambda)$ maps $L^2_\comp$ to $H^1_\loc$ for $\lambda \in \Omega_*$, and $\EE_\gamma$ to $H^1_\loc$ for $\lambda \in \tOmega_*$. This completes the proof.
\end{proof}

\section{Tilted domain walls}\label{sec:tilted}

We consider now domain walls whose interface is an arbitrary line through the origin, instead of the horizontal line $\R e_1$ of \Cref{def:domain-wall}. We label such a line by a unit vector $e$ spanning it.

\subsection{The tilted interface}\label{subsec:tilted}

Throughout this section we fix a unit vector $e \in \R^2$. We denote by $e^\perp$ the rotation of $e$ by $\pi/2$, and by $O_e$ the rotation taking $e_1$ to $e$:
\begin{equation}\label{eq:rotation}
    e = \begin{bmatrix} \cos\theta \\ \sin\theta \end{bmatrix},
    \qquad
    e^\perp \defeq \begin{bmatrix} -\sin\theta \\ \cos\theta \end{bmatrix},
    \qquad
    O_e \defeq \begin{bmatrix} \cos\theta & -\sin\theta \\ \sin\theta & \cos\theta \end{bmatrix} \in SO(2),
\end{equation}
so that $O_e e_1 = e$ and $O_e e_2 = e^\perp$, where $e_1, e_2$ is the canonical basis of $\R^2$. The angle $\theta$ is determined by $e$ modulo $2\pi$, while $e^\perp$ and $O_e$ depend on $e$ alone.

\begin{definition}\label{def:tilted}
Given a horizontal domain wall $\kappa_*$ with profile $k_*$ and a unit vector $e \in \R^2$, the tilted domain wall in the direction $e$ is $\kappa_e(x) \defeq k_*\left(x\cdot e^\perp\right)$, and the tilted Dirac operator is
\begin{equation}\label{eq:dirac-e}
    \Di_e \defeq D_1\sigma_1 + D_2\sigma_2 + \kappa_e(x)\sigma_3
    = D_1\sigma_1 + D_2\sigma_2 + k_*\left(\cos\theta\, x_2 - \sin\theta\, x_1\right)\sigma_3 ,
\end{equation}
acting on $L^2$ with domain $H^1$.
\end{definition}

The mass now changes sign across the line $\R e$, rather than across $\R e_1$; see \Cref{fig:tilted}. 

\begin{figure}[ht]
\centering
\begin{minipage}[c]{0.52\textwidth}
    \centering
    \includegraphics[width=\linewidth]{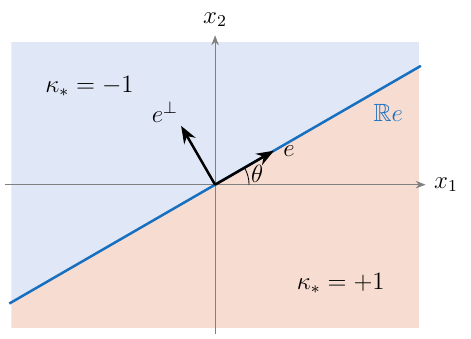}
\end{minipage}\hfill
\begin{minipage}[c]{0.45\textwidth}
\caption{The tilted domain wall of \Cref{def:tilted}. The mass $\kappa_e$ equals $-1$ on the shaded blue side $\{x\cdot e^\perp \ge 1\}$ and $+1$ on the shaded red side $\{x\cdot e^\perp \le -1\}$. The edge states \eqref{eq:psi-lambda-e} are supported near $\R e$ and propagate in the direction $e$.}
\label{fig:tilted}
\end{minipage}
\end{figure}

\begin{proposition}\label{prop:tilted}
For $e \in \R^2$ a unit vector, let $S_e$ be the operator on $L^2$ given by
\begin{equation}\label{eq:U-e}
    \left(S_e f\right)(x) \defeq U_e\, f\left(O_e x\right), \qquad U_e \defeq \begin{bmatrix} e^{i\theta/2} & 0 \\ 0 & e^{-i\theta/2}\end{bmatrix} \in SU(2).
\end{equation}
Then $S_e$ is unitary, maps $H^1$ onto itself, and
\begin{equation}\label{eq:tilt-conjugation}
    \Di_e = S_e^{-1}\, \Di_*\, S_e .
\end{equation}
\end{proposition}

For the proof, see \cite[Lemma 2.1]{BalEtAl23}. We can now transfer everything proved in \S\ref{sec:unperturbed}--\ref{sec:resolvent} for $\Di_*$ to $\Di_e$. Below, $\phi_*$ and $K_*$ are as in \Cref{lem:zero-mode}, and $\psi_{\lambda,*} = \psi_{\lambda,e_1}$ is the edge state \eqref{eq:psi-lambda}.

\begin{corollary}\label{cor:tilted}
Let $\kappa_*$ be a horizontal domain wall and let $e \in \R^2$ be a unit vector. Then:
\begin{enumerate}[label=(\alph*)]
    \item $\Di_e$ is self-adjoint on $L^2$ with domain $H^1$, and $\spec(\Di_e) = \spec(\Di_*)$.
    \item For every $\lambda \in \R$, the function $\psi_{\lambda,e} \defeq S_e^{-1}\psi_{\lambda,*}$ is given by
    \begin{equation}\label{eq:psi-lambda-e}
        \psi_{\lambda,e}(x) = e^{i\lambda\, x\cdot e}\, \phi_e\left(x \cdot e^\perp\right),
        \qquad
        \phi_e(s) \defeq U_e^*\phi_* = c_* \begin{bmatrix} e^{-i\theta/2} \\ e^{i\theta/2}\end{bmatrix} e^{K_*(s)},
    \end{equation}
    and it satisfies $\Di_e \psi_{\lambda,e} = \lambda\, \psi_{\lambda,e}$ together with
    \begin{equation}\label{eq:tilted-chirality}
        \left(\sigma\cdot e\right)\phi_e = \phi_e , \qquad \sigma\cdot e \defeq \cos \theta \sigma_1 + \sin \theta \sigma_2, \qquad e = \begin{bmatrix} \cos\theta \\ \sin\theta \end{bmatrix}.
    \end{equation}
\end{enumerate}
\end{corollary}

The proof is a direct consequence of \Cref{prop:tilted}. The same conjugation transports the flux computation of \Cref{lem:edge-flux} to an arbitrary direction:

\begin{lemma}\label{lem:edge-flux-tilted}
Let $\lambda \in \R$, let $\chi$ be as in \Cref{lem:edge-flux}, and let $\zeta \in \R^2$ be a unit vector with $c \defeq \zeta\cdot e > 0$. Set
\begin{equation}\label{eq:edge-flux-w-e}
    w_e(x) \defeq \chi\left(x\cdot \zeta\right) e^{i\lambda\, x\cdot e}\, \phi_e\left(x \cdot e^\perp\right) .
\end{equation}
For $R >0$, we have:
\begin{equation}\label{eq:tilted-edge-flux}
    \int_{\partial B_R(0)} J\left(w_e,w_e\right)\cdot n \ d\ell = 1+O\left(e^{-cR/2}\right).
\end{equation}
\end{lemma}

\begin{proof} We first establish how $J(\varphi,\psi)$ changes under the gauge transformations of \Cref{prop:tilted}. We have $\left(S_e^{-1}f\right)(x) = U_e^* f\left(O_e^\top x\right)$, so for $\varphi,\psi \in H^1_\loc$ and $j=1,2$,
\begin{equation}
    J_j\left(S_e^{-1}\varphi, S_e^{-1}\psi\right)(x)
    = \ove{\varphi\left(O_e^\top x\right)}^\top \left(U_e \sigma_j U_e^*\right) \psi\left(O_e^\top x\right) .
\end{equation}
A direct computation from \eqref{eq:pauli} and \eqref{eq:U-e} gives
\begin{equation}
    U_e \sigma_1 U_e^* = \cos\theta\, \sigma_1 - \sin\theta\, \sigma_2 , \qquad
    U_e \sigma_2 U_e^* = \sin\theta\, \sigma_1 + \cos\theta\, \sigma_2 ;
\end{equation}
that is, $U_e \sigma_j U_e^* = \sum_k \left(O_e\right)_{jk}\sigma_k$ with $O_e$ as in \eqref{eq:rotation}. Therefore
\begin{equation}\label{eq:current-rotates}
    J\left(S_e^{-1}\varphi, S_e^{-1}\psi\right)(x) = O_e\, J(\varphi,\psi)\left(O_e^\top x\right) .
\end{equation}
Let now $w_*$ be as in \eqref{eq:edge-flux-w}, with cut-off direction $O_e^\top\zeta$; this is legitimate because $O_e^\top\zeta\cdot e_1 = \zeta\cdot O_ee_1 = \zeta\cdot e = c > 0$. Since $O_e^\top x \cdot O_e^\top \zeta = x\cdot\zeta$, $O_e^\top x \cdot e_1 = x \cdot e$ and $O_e^\top x \cdot e_2 = x\cdot e^\perp$, and since $U_e^*\phi_* = \phi_e$ by \eqref{eq:psi-lambda-e}, we have $w_e = S_e^{-1}w_*$. The change of variable $x = O_e y$ maps $\partial B_R(0)$ to itself, preserves the arclength, and takes the outer normal $n(y)$ to $n(x) = O_e n(y)$; so \eqref{eq:current-rotates} and the orthogonality of $O_e$ give
\begin{equation}
    \int_{\partial B_R(0)} J\left(w_e,w_e\right)\cdot n \ d\ell
    = \int_{\partial B_R(0)} O_e J(w_*,w_*)(y) \cdot O_e n(y) \ d\ell(y)
    = \int_{\partial B_R(0)} J(w_*,w_*)\cdot n \ d\ell .
\end{equation}
The conclusion now follows from \Cref{lem:edge-flux}.
\end{proof}

There is also a notion of edge projection: we set
\begin{equation}\label{eq:tilted-projection}
    P_e \defeq S_e^{-1}\,P_*\,S_e , \qquad P_e^\perp \defeq \Id - P_e , \qquad \HH_e \defeq \ran\left(P_e\right) = S_e^{-1}\HH_* .
\end{equation}
Since $S_e$ is unitary, $P_e$ is the orthogonal projection of $L^2$ onto $\HH_e$; and by \Cref{lem:decomposition} together with \eqref{eq:tilt-conjugation}, it commutes with $\Di_e$. Unwinding \eqref{eq:U-e} and using $U_e^*\phi_* = \phi_e$ gives
\begin{equation}\label{eq:tilted-projection-formula}
    \left(P_e f\right)(x) = a_e\left(x\cdot e\right)\, \phi_e\left(x\cdot e^\perp\right) ,
\end{equation}
with $a_e$ as in \eqref{eq:tilted-resolvent-structure}. \Cref{lem:projection-weight} transfers to the tilted setting:

\begin{lemma}\label{lem:projection-weight-tilted}
Let $e, \zeta \in \R^2$ be unit vectors and let $0<\gamma<1$. If $e^{\gamma\left|x\cdot\zeta\right|}f \in L^2$, then $e^{\gamma\left|x\cdot\zeta\right|}P_ef$ and $e^{\gamma\left|x\cdot\zeta\right|}P_e^\perp f$ belong to $L^2$.
\end{lemma}

The proof follows directly from applying \Cref{lem:projection-weight} in the direction $O_e^{-1} \zeta$, then using the conjugation relation between $\Di_e$ and $\Di_*$ -- which transfers to $P_e$ and $P_*$. 
Likewise, we obtain the analogue of the analytic continuation of $(\Di_*-\lambda)^{-1}$ to the tilted case -- again we skip the details: 

\begin{theorem}\label{cor:tilted-resolvent}
Let $\kappa_*$ be a horizontal domain wall and let $e \in \R^2$ be a unit vector. The family of operators $\left(\Di_e-\lambda\right)^{-1} : L^2_\comp \rightarrow H^1_\loc$, defined initially for $\lambda \in \C^+$, continues analytically to a family of operators
\begin{equation}\label{eq:tilted-resolvent}
    R_e(\lambda)  \ : \ L^2_\comp \longrightarrow H^1_\loc, \qquad \lambda \in \Omega_* .
\end{equation}
Moreover, for $\lambda \in \tOmega_*$ and $\gamma$ such that $|\Im\lambda| < \gamma \le \gamma_*(\lambda)$, the operator $R_e(\lambda)$ extends by density from $\EE_\gamma$ to $H^1_\loc$ and we have the spatial asymptotics
\begin{equations}\label{eq:tilted-resolvent-structure}
    R_e(\lambda) f(x) = i\int_{-\infty}^{x\cdot e} e^{i\lambda\left(x\cdot e-s\right)}a_e(s)\,ds \cdot \phi_e\left(x\cdot e^\perp\right) \; + \; \EE_{\gamma/2}, \qquad f \in \EE_\gamma,
\\
    a_e(s) \defeq \int_\R \ove{\phi_e(\tau)}^{\top}\, f\left(s\,e + \tau\,e^\perp\right) d\tau .
\end{equations}
\end{theorem}

\section{Bent domain walls}\label{sec:bent}

\subsection{Definitions} We now let the interface turn a corner. Throughout this subsection we fix two unit vectors $e_- = e_1$ and $e_+ \neq -e_1$. We set $\theta_- = 0$, denote by $\theta_+ \in (-\pi,\pi)$ the angle between $e_-$ and $e_+$, and define
\begin{equation}
    \theta_0 \defeq \dfrac{\theta_+}{2} \in \left(-\dfrac\pi2,\dfrac\pi2\right), \qquad e_0 \defeq e_{\theta_0} = \begin{bmatrix} \cos\theta_0 \\ \sin\theta_0 \end{bmatrix}.
\end{equation}
All identities between angles below are understood modulo $2\pi$.

We consider tilted domain walls of the form
\begin{equation}
    \kappa_\pm(x) = k_\pm\left(x \cdot e_\pm^\perp\right),
\end{equation}
where $k_\pm \in L^\infty(\R)$ are such that $k_\pm(x_2) = -1$ for $x_2 \geq 1$ and $k_\pm(x_2) = 1$ for $x_2 \leq -1$.

\begin{definition}\label{def:bent}
A domain wall bent by the angle $\theta_+$, with asymptotics $\kappa_\pm$, is a function $\kappa \in L^\infty(\R^2,\R)$ such that for some $R>0$,
\begin{equation}\label{eq:bent}
    \kappa(x) = \kappa_\pm(x) \qquad\text{whenever}\qquad |x| > R \quad\text{and}\quad \pm\, x\cdot e_0 \, > \, 0 .
\end{equation}
\end{definition}

\begin{figure}[b]
\centering
\begin{minipage}[c]{0.52\textwidth}
    \centering
    \includegraphics[width=\linewidth]{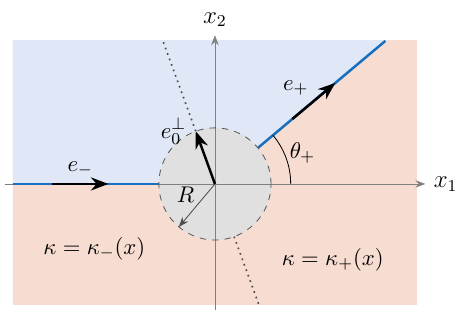}
\end{minipage}\hfill
\begin{minipage}[c]{0.45\textwidth}
\caption{A bent domain wall with angle $\theta_+$. For $|x| > R$, $\kappa$ agrees with $\kappa_-$ on $\left\{x\cdot e_0<0\right\}$ and with $\kappa_+$ on $\left\{x\cdot e_0>0\right\}$. It equals $-1$ in the blue area and $+1$ in the red area. There is no requirement in an at-worst bounded region (gray).}
\label{fig:bent}
\end{minipage}
\end{figure}

For instance, if $S = \{ x \in \R^2 : \theta_+ < \arg x \leq \pi\}$, then $\1_{S^c} - \1_S$, as well as
\begin{equation}\label{eq:bent-examples}
    \kappa = \rho * (\1_{S^c} - \1_S) + V,
\end{equation}
with $\rho \in C^\infty_0(\R^2)$ of integral one, and $V \in C^\infty_0(\R^2)$, are bent domain walls in the sense of \Cref{def:bent}; see \Cref{fig:intro-bent} for a depiction. We give several additional examples in \Cref{subsec:numerics-walls2}.

Here, we study the bent Dirac operator
\begin{equation}\label{eq:dirac-bent}
    \Di \defeq D_1\sigma_1 + D_2\sigma_2 + \kappa(x)\sigma_3 ,
\end{equation}
where $\kappa$ satisfies \Cref{def:bent}. The Dirac operator $\Di$ is selfadjoint on $L^2$ with domain $H^1$. Its two models at infinity are the operators
\begin{equation}\label{eq:models}
    \Di_\pm \defeq D_1\sigma_1 + D_2\sigma_2 + \kappa_\pm(x)\,\sigma_3 .
\end{equation}

According to \Cref{sec:tilted}, the functions $\phi_\pm$ that seed the edge states of $\Di_\pm$ are given by
\begin{equation}\label{eq:phi-pm}
    \phi_\pm(s) = c_\pm \begin{bmatrix} e^{-i\theta_\pm/2} \\ e^{i\theta_\pm/2}\end{bmatrix} e^{\int_0^{s} k_\pm},
\end{equation}
where $c_\pm$ are normalizing constants so that  $\| \phi_\pm \| = 1$; we write $M_\pm \defeq \norm{e^{|x_2|}\phi_\pm}_\infty$, the analogue for $\phi_\pm$ of the constant $M_*$ of \Cref{lem:zero-mode}. Recall that 
\begin{equations}
    E \defeq \min(E_+,E_-), \qquad 
    E_\pm \defeq \min\Big(\spec\big(\Di_\perp^\pm\big)\cap (0,+\infty) \Big),
    \qquad
    \Di_\perp^\pm \defeq D_2\sigma_2 + k_\pm(x_2)\,\sigma_3.
\end{equations}
Specializing \Cref{cor:tilted-resolvent} to $e = e_\pm$ gives the structure of the two model resolvents $R_\pm(\lambda)$: these continue analytically for
\begin{equation}\label{eq:def-gamma-pm}
    \lambda \in \widetilde{\Omega}_\pm \defeq \left\{ \lambda \in \C : |\lambda| < E_\pm, |\Im \lambda| < \gamma_\pm(\lambda) \right\}, \qquad \gamma_\pm(\lambda) \defeq \min\left( \dfrac{1}{2M_\pm^2}, \dfrac{E_\pm^2-|\lambda|^2}{20}\right)
\end{equation}
and moreover, for $\lambda \in \tOmega_\pm$ and $\gamma$ such that $|\Im\lambda| < \gamma \le \gamma_\pm(\lambda)$, $R_\pm(\lambda)$ is continuous from $\EE_\gamma$ to $H^1_\loc$, with
\begin{equations}\label{eq:2z}
    R_\pm(\lambda) f = i\,e^{i\lambda\, x\cdot e_\pm}\left(\int_{-\infty}^{x\cdot e_\pm} e^{-i\lambda s}\,a_\pm(s)\,ds\right)\phi_\pm\left(x\cdot e_\pm^\perp\right) + \; \EE_{\gamma/2}, \qquad f \in \EE_\gamma,
    \\
    a_\pm(s) \defeq \int_\R \ove{\phi_\pm(\tau)}^{\top}\, f\left(s\,e_\pm + \tau\,e_\pm^\perp\right) d\tau.
\end{equations}

\subsection{Resolvent of $\Di$}\label{subsec:mero}
Our first goal is to continue $(\Di-\lambda)^{-1}$ meromorphically across $(-E,E)$. Recall that $\gamma_\pm(\lambda)$ were defined in \eqref{eq:def-gamma-pm} and set
\begin{equations}\label{eq:def_gamma}
    \tOmega \defeq \left\{ \lambda \in \C \ : \ |\lambda| < E, \quad \left|\Im\lambda\right| < \left(1-\left|\sin\te_0\right|\right)\gamma(\lambda) \right\} \subset \tOmega_+ \cap \tOmega_-,
    \\
    \gamma(\lambda) \defeq \min\big( \gamma_-(\lambda), \gamma_+(\lambda) \big).
\end{equations}
The set $\tOmega$ is a neighborhood of $(-E,E)$ in $\C$. See \Cref{fig:tomega}.

\begin{figure}[ht]
\centering
\begin{minipage}[c]{0.58\textwidth}
    \centering
    \includegraphics[width=\linewidth]{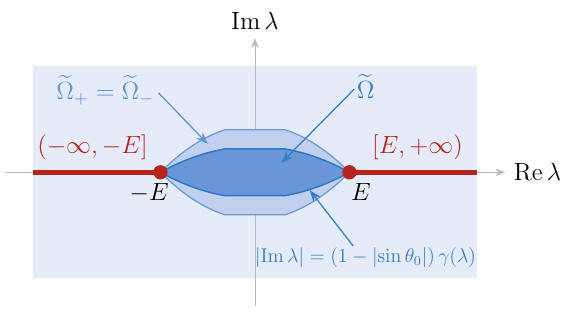}
\end{minipage}\hfill
\begin{minipage}[c]{0.39\textwidth}
\caption{The set $\tOmega$ defined in \eqref{eq:def_gamma}, with $\tOmega_+ = \tOmega_-$ for simplicity.}
\label{fig:tomega}
\end{minipage}
\end{figure}

\begin{theorem}\label{thm:mero}
The family of operators $\left(\Di-\lambda\right)^{-1} : L^2_\comp \rightarrow H^1_\loc$, originally defined for $\Im\lambda>0$, continues meromorphically to a family of operators
\begin{equation}\label{eq:mero}
    R(\lambda) : L^2_\comp \longrightarrow H^1_\loc, \qquad \lambda \in \tOmega,
\end{equation}
with poles of finite rank. 

Moreover, for $\lambda \in \tOmega$ and $\gamma > 0$ such that
\begin{equation}\label{eq:mero-gamma}
    \left|\Im\lambda\right| \ < \ \gamma \ < \ \left(1-\left|\sin\te_0\right|\right)\gamma(\lambda) ,
\end{equation}
the operator $R(\lambda)$ extends by density from $\EE_\gamma$ to $H^1_\loc$.
\end{theorem}

Poles $\lambda_0 \in \tOmega$ of $R(\lambda)$ are called resonances  of $\Di$. Resonances of $\Di$ on the real axis are necessarily eigenvalues of $\Di$, as can be seen from the bound $\|R(\lambda)\| \leq |\Im \lambda|^{-1}, \Im\lambda > 0$, coming from selfadjointness of $\Di$ -- see \cite[Theorem 4.18]{DZ}.

\subsection{$(\Di-\lambda)^{-1}$ continues meromorphically}
Let $\chi \in C^\infty(\R)$ be equal to $1$ on $[1,+\infty)$ and to $0$ on $(-\infty,-1]$, and set, for $\lambda \in \tOmega$:
\begin{equations}\label{eq:parametrix}
    Q(\lambda) \defeq \sum_\pm \chi_\pm\, R_\pm(\lambda) : \EE_\gamma \longrightarrow L^2_\loc, \qquad \lambda \in \tOmega ,
    \\
    \chi_+(x) \defeq \chi\left(e_0 \cdot x\right), \qquad \chi_-(x) \defeq 1-\chi_+(x).
\end{equations}
Heuristically, $Q(\lambda)$ should model the spatial asymptotics of the resolvent $(\Di-\lambda)^{-1}$. Using $\Di = \Di_\pm + \left(\kappa-\kappa_\pm\right)\sigma_3$, we compute:
\begin{align}\label{eq:leibniz}
    \left(\Di-\lambda\right)\chi_\pm R_\pm(\lambda)
    & = \chi_\pm\left(\Di_\pm - \lambda\right) R_\pm(\lambda)
    + \chi_\pm\left(\kappa - \kappa_\pm\right)\sigma_3 R_\pm(\lambda)
    + \left[\Di,\chi_\pm\right] R_\pm(\lambda)
    \\
    & = \chi_\pm + \chi_\pm\left(\kappa - \kappa_\pm\right)\sigma_3 R_\pm(\lambda) + (\sigma \cdot D_x \chi_\pm) R_\pm(\lambda)
    \\
    & = \chi_\pm + \chi_\pm\left(\kappa - \kappa_\pm\right)\sigma_3 R_\pm(\lambda) \pm (\sigma \cdot D_x \chi_+) R_\pm(\lambda).
\end{align}
Thanks to $\chi_++\chi_-=1$, we deduce that:
\begin{equations}\label{eq:parametrix-error}
\begin{aligned}
    \left(\Di - \lambda\right) Q(\lambda) &= \Id + K(\lambda),
    &\qquad K(\lambda) &\defeq K_1(\lambda) + K_2(\lambda),
    \\
    K_1(\lambda) &\defeq \sum_\pm V_\pm R_\pm(\lambda),
    &\qquad V_\pm &\defeq \chi_\pm\left(\kappa - \kappa_\pm\right)\sigma_3,
    \\
    K_2(\lambda) &\defeq W \big( R_+(\lambda) - R_-(\lambda)\big),
    &\qquad W &\defeq \sigma \cdot D_x\chi_+ .
\end{aligned}
\end{equations}
Our strategy is to use analytic Fredholm theory, see e.g. \cite[Appendix C]{DZ}. In particular, we must justify that $K(\lambda)$ is a compact operator on $\EE_\gamma$.

\begin{lemma}\label{lem:cpt-gain} Let $\lambda \in \widetilde{\Omega}_+ \cap \widetilde{\Omega}_-$ and $\gamma > 0$ be such that:
\begin{equations}\label{eq:cpt-parameters}
    |\lambda| < E, \qquad \left|\Im\lambda\right| < \gamma < \gamma(\lambda)\left(1-\left|\sin\te_0\right|\right).
\end{equations}
    Then there exists $\rho > \gamma$ such that $K_2(\lambda)$ maps $\EE_\gamma$ to $\EE_\rho$.
\end{lemma}

\begin{proof} Throughout the proof, $\lambda$ and $\gamma$ are parameters satisfying \eqref{eq:cpt-parameters}.

    \textbf{1.} We first prove a useful decomposition for $K_2(\lambda)$. By adapting the decomposition of the untilted resolvent \eqref{eq:2a} to the tilted case, we may expand
    \begin{equation}
        R_\pm(\lambda) = A_\pm(\lambda) + B_\pm(\lambda), \qquad A_\pm(\lambda) \defeq R_\pm(\lambda) P_\pm, \qquad B_\pm(\lambda) \defeq R_\pm(\lambda) P_\pm^\perp. 
    \end{equation}
It follows that 
\begin{equation}\label{eq:cpt-AB}
    K_2(\lambda) = W A_+(\lambda) - W A_-(\lambda) + W \big( B_+(\lambda) - B_-(\lambda) \big).
\end{equation}

    Because $B_\pm(\lambda)$ has range in $\HH_\pm^\perp$, we have $P_\pm^\perp B_\pm(\lambda) = B_\pm(\lambda)$, and inserting $\Id = P_-+P_-^\perp$ on the right of $B_+(\lambda)$ and $\Id = P_++P_+^\perp$ on the left of $B_-(\lambda)$,
\begin{equation}\label{eq:cpt-B-split}
    B_+(\lambda)-B_-(\lambda) = B_+(\lambda)P_- \; - \; P_+ B_-(\lambda) \; + \; \left(B_+(\lambda)P_-^\perp - P_+^\perp B_-(\lambda)\right) .
\end{equation}
We have the identity
\begin{equation}\label{eq:cpt-resolvent-identity}
    B_+(\lambda)P_-^\perp - P_+^\perp B_-(\lambda) \ = \ B_+(\lambda)\,Z\,B_-(\lambda) , \qquad Z \defeq \left(\kappa_--\kappa_+\right)\sigma_3 = \Di_- - \Di_+.
\end{equation}
Indeed, for $\Im\lambda>0$ the operators $R_\pm(\lambda) = \left(\Di_\pm-\lambda\right)^{-1}$ are bounded on $L^2$ and commute with $P_\pm$, so that the resolvent identity $R_+(\lambda)-R_-(\lambda) = R_+(\lambda)ZR_-(\lambda)$ gives
\begin{equation}
    B_+(\lambda)P_-^\perp - P_+^\perp B_-(\lambda) = P_+^\perp\left(R_+(\lambda)-R_-(\lambda)\right)P_-^\perp = P_+^\perp R_+(\lambda)\, Z\, R_-(\lambda)P_-^\perp = B_+(\lambda) Z B_-(\lambda) .
\end{equation}
Therefore \eqref{eq:cpt-resolvent-identity} holds for $\Im\lambda>0$; by analyticity, it holds for $|\lambda| < E$.

Combining \eqref{eq:cpt-AB}, \eqref{eq:cpt-B-split} and \eqref{eq:cpt-resolvent-identity}, we obtain the decomposition
\begin{equations}\label{eq:cpt-decomposition}
\begin{aligned}
    K_2(\lambda) &= \underbrace{W A_+(\lambda)}_{(1)} \; - \; \underbrace{W A_-(\lambda)}_{(2)} \; + \; \underbrace{W B_+(\lambda)\,Z\,B_-(\lambda)}_{(3)}
    \\[4pt]
    &\qquad + \; \underbrace{W B_+(\lambda)P_-}_{(4)} \; - \; \underbrace{W P_+B_-(\lambda)}_{(5)} .
\end{aligned}
\end{equations}
We prove in steps \textbf{3}--\textbf{6} below that each of $(1)$--$(5)$ is bounded from $\EE_\gamma$ to $\EE_\rho$, for some $\rho > \gamma$.

\textbf{2.} We record here a few weighted bounds on the operators $P_\pm, B_\pm(\lambda)$. \Cref{lem:projection-weight-tilted} applied to $P_\pm = P_{e_\pm}$ shows that the operators
\begin{equation}\label{eq:cpt-projection-bounds}
    P_\pm, P_\pm^\perp  \ : \ e^{-\mu |\zeta \cdot x|} L^2 \rightarrow e^{-\mu |\zeta \cdot x|} L^2 , \qquad |\zeta| = 1, \qquad \mu \in (0,1)
\end{equation}
are bounded. Together with (the tilted version of) \Cref{prop:agmon}, this implies  that
\begin{equation}\label{eq:cpt-agmon}
    B_\pm(\lambda) \ : \ e^{-\epsi |\zeta \cdot x|} L^2 \rightarrow e^{-\epsi |\zeta \cdot x|} L^2 , \qquad 0< \epsi \leq 2\gamma(\lambda), \qquad  |\zeta| = 1 
\end{equation}
is bounded.

Moreover, the operator
\begin{equation}\label{eq:cpt-projection-transverse}
    P_\pm \ : \ L^2 \ \longrightarrow \ e^{-\mu\left|x\cdot e_\pm^\perp\right|} L^2 , \qquad \mu \in (0,1)
\end{equation}
is bounded. Indeed,
\begin{equation}
    P_\pm f(x) = a_\pm\left(x\cdot e_\pm\right)\, \phi_\pm\left(x\cdot e_\pm^\perp\right) , \qquad a_\pm(s) \defeq \int_\R \ove{\phi_\pm(\tau)}^{\top} f\left(s\,e_\pm + \tau\, e_\pm^\perp\right) d\tau .
\end{equation}
Since $\norm{\phi_\pm}_{L^2(\R)} = 1$ and $P_\pm$ is an orthogonal projection, $\norm{a_\pm}_{L^2(\R)} = \norm{P_\pm f} \leq \norm{f}$. Thus
\begin{equation}
    \big\| e^{\mu\left|x\cdot e_\pm^\perp\right|} P_\pm f \big\| \ = \ \norm{a_\pm}_{L^2(\R)}\ \big\| e^{\mu\left|x\cdot e_\pm^\perp\right|}\phi_\pm \big\|_{L^2(\R)} \ \leq \ \big\| e^{\mu\left|x\cdot e_\pm^\perp\right|}\phi_\pm \big\|_{L^2(\R)}\, \norm{f} ,
\end{equation}
and $\big\| e^{\mu|\tau|}\phi_\pm \big\|_{L^2(\R)} < +\infty$ because $\mu < 1$ and $\left|\phi_\pm(\tau)\right| \leq M_\pm e^{-|\tau|}$ by \eqref{eq:phi-pm}. This proves \eqref{eq:cpt-projection-transverse}.

\textbf{3.} We work on the terms $W A_\pm(\lambda)$ in \eqref{eq:cpt-decomposition}. Heuristically, these map $\EE_\gamma$ to $\EE_\rho$ for some $\rho > \gamma$ because they add decay in the $e_\pm^\perp$-direction (due to $P_\pm$) and in the $e_0$-direction (due to $W$), and $\left(e_\pm^\perp, e_0\right)$ form a basis of $\R^2$. Let us be more specific. We recall from \eqref{eq:phi-pm} and \eqref{eq:edge-integral-bound} that
\begin{equations}\label{eq:cpt-edge-factors}
    \left|\phi_\pm\left(x\cdot e_\pm^\perp\right)\right| \ \le \ C\,e^{-\left|x\cdot e_\pm^\perp\right|},
\qquad
\left|\int_{-\infty}^{x\cdot e_\pm} e^{-i\lambda s}a_\pm(s)\,ds\right| \ \le \ \dfrac{\norm{f}_{\EE_\gamma}}{\sqrt{\gamma-\left|\Im\lambda\right|}}.
\end{equations}
This leads to:
\begin{align}
    \big|A_\pm(\lambda)f(x)\big| 
    & = \left| e^{i\lambda\, x\cdot e_\pm}\left(\int_{-\infty}^{x\cdot e_\pm} e^{-i\lambda s}\,a_\pm(s)\,ds\right)\phi_
    \pm\left(x\cdot e^\perp_\pm\right)\right|
    \\
    & \leq \ \dfrac{C\,\norm{f}_{\EE_\gamma}}{\sqrt{\gamma-\left|\Im\lambda\right|}}\; e^{\gamma|x| \, - \, \left|x\cdot e_\pm^\perp\right|} . \label{eq:cpt-edge-A-bound}
\end{align}

The function $W$ is supported in the strip $|x \cdot e_0| \leq 1$, and so $W(x) = O(e^{-|x \cdot e_0|})$. Combined with \eqref{eq:cpt-edge-A-bound}, this produces:
\begin{equation}\label{eq:cpt-amplitude}
    \big|W(x) A_\pm(\lambda)f(x)\big| 
     \leq \ \dfrac{C\,\norm{f}_{\EE_\gamma}}{\sqrt{\gamma-\left|\Im\lambda\right|}}\; e^{\gamma|x| \, - \, \left|x\cdot e_\pm^\perp\right| \, - \, \left|x\cdot e_0\right|} .
\end{equation}

We now study $\left|x\cdot e_\pm^\perp\right| + \left|x\cdot e_0\right|$. We have
\begin{equation}\label{eq:cpt-edge-geometry}
    \left|x\cdot e_\pm^\perp\right| + \left|x\cdot e_0\right| \ \ge \ \sqrt{\left(x\cdot e_\pm^\perp\right)^2 + \left(x\cdot e_0\right)^2} = \sqrt{x^\top M_\pm\, x}, \qquad M_\pm \defeq e_\pm^\perp\left(e_\pm^\perp\right)^\top + e_0\,e_0^\top .
\end{equation}
Set $c_\pm \defeq e_\pm^\perp \cdot e_0 = \mp \sin \te_0$. The matrix $M_\pm$ maps $e_\pm^\perp$ to $e_\pm^\perp + c_\pm\,e_0$ and $e_0$ to $c_\pm\,e_\pm^\perp + e_0$. In that basis $M_\pm$ is represented by
\begin{equation}\label{eq:cpt-edge-matrix}
    \begin{bmatrix} 1 & c_\pm \\ c_\pm & 1 \end{bmatrix} ,
    \qquad\text{whose eigenvalues are}\qquad
    1 \pm c_\pm .
\end{equation}
The smallest eigenvalue of $M_\pm$ is therefore $1-\left|\sin\theta_0\right| \in (0,1]$, hence
\begin{equation}\label{eq:cpt-edge-sqrt}
    \left|x\cdot e_\pm^\perp\right| + \left|x\cdot e_0\right| \ \ge \sqrt{x^\top M_\pm\, x} \ \ge \ \sqrt{1-\left|\sin\theta_0\right|}\ |x| \ \ge \ \left(1-\left|\sin\theta_0\right|\right)|x|.
\end{equation}

Combining \eqref{eq:cpt-amplitude} and \eqref{eq:cpt-edge-sqrt}, we obtain 
\begin{equation}\label{eq:cpt-edge-radial}
    \left|W(x)\,A_\pm(\lambda)f(x)\right| \ \le \ \dfrac{C\,\norm{f}_{\EE_\gamma}}{\sqrt{\gamma-\left|\Im\lambda\right|}}\; e^{-\left(1-\left|\sin\theta_0\right|-\gamma\right)|x|} .
\end{equation}
In particular, $WA_\pm(\lambda)$ map $\EE_\gamma$ to $\EE_\rho$ when $\rho < 1 -\left|\sin\theta_0\right|-\gamma$.

\textbf{4.} We now work on the term $W B_+(\lambda) P_-$ in \eqref{eq:cpt-decomposition}. This term should also map $\EE_\gamma$ to $\EE_\rho$ for the same heuristic reason as the terms $(1)$ and $(2)$. By \eqref{eq:cpt-projection-transverse} with $\mu = \tfrac12$, and by \eqref{eq:cpt-agmon}, the operators
\begin{equation}\label{eq:cpt-term4-projection}
\begin{aligned}
    P_- \ &: \ L^2 \ \longrightarrow \ e^{-\frac12\left|x\cdot e_-^\perp\right|}L^2 \ \subset \ e^{-\gamma(\lambda) \left|x\cdot e_-^\perp\right|}L^2 ,
    \\
    B_+(\lambda) \ &: \ e^{-\gamma(\lambda) \left|x\cdot e_-^\perp\right|}L^2 \ \longrightarrow \ e^{-\gamma(\lambda) \left|x\cdot e_-^\perp\right|}L^2
\end{aligned}
\end{equation}
are bounded. Therefore, since $W$ is supported in the strip $|x \cdot e_0| \leq 1$ by \eqref{eq:parametrix} and \eqref{eq:parametrix-error}, we have:
\begin{equation}
    W B_+(\lambda) P_- \ : \ L^2 \ \longrightarrow \ e^{-|x\cdot e_0|-\gamma(\lambda) \left|x\cdot e_-^\perp\right|}L^2.
\end{equation}
By \eqref{eq:cpt-edge-sqrt}, and because $\gamma(\lambda) < 1$ by \eqref{eq:def_gamma}, we have
\begin{equation}
    |x\cdot e_0|+\gamma(\lambda) \left|x\cdot e_-^\perp\right| \geq \gamma(\lambda) \big( |x\cdot e_0| + \left|x\cdot e_-^\perp\right|\big) \geq \gamma(\lambda) (1-|\sin \te_0|) |x|.
\end{equation}
Hence, for $\rho \leq \gamma(\lambda) (1-|\sin \te_0|)$, the operator $W B_+(\lambda) P_-$ maps $L^2$ to $\EE_\rho$, hence $\EE_\gamma$ to $\EE_\rho$.

\textbf{5.} We now work on the term $W P_+ B_-(\lambda)$ in \eqref{eq:cpt-decomposition}. Because $|\lambda| < E \le E_-$, $\lambda$ is not in the spectrum of $\Di_-\big|_{\HH_-^\perp}$; in particular, $B_-(\lambda)$ is bounded on $L^2$.

It therefore suffices to show that $W P_+$ is bounded from $L^2$ to $\EE_\rho$ for some $\rho > \gamma$. By \eqref{eq:cpt-projection-transverse} with $\mu = \tfrac12$, and because $W$ is supported in the strip $\left|x\cdot e_0\right| \leq 1$, the operators
\begin{equation}\label{eq:cpt-term5-projection}
\begin{aligned}
    P_+ \ &: \ L^2 \ \longrightarrow \ e^{-\frac12\left|x\cdot e_+^\perp\right|}L^2 ,
    \\
    W \ &: \ e^{-\frac12\left|x\cdot e_+^\perp\right|}L^2 \ \longrightarrow \ e^{-\frac12\left(\left|x\cdot e_+^\perp\right| \, + \, \left|x\cdot e_0\right|\right)}L^2
\end{aligned}
\end{equation}
are bounded. By \eqref{eq:cpt-edge-sqrt}, we have
\begin{equation}
    \dfrac12\left(\left|x\cdot e_+^\perp\right| + \left|x\cdot e_0\right|\right) \ \geq \ \dfrac{1-\left|\sin\te_0\right|}{2}\,|x| .
\end{equation}
For $\rho \leq \frac12\left(1-\left|\sin\te_0\right|\right)$, we deduce that $W P_+ B_-(\lambda) : L^2 \rightarrow \EE_\rho$ is bounded. In particular, $W P_+ B_-(\lambda)$ maps $\EE_\gamma$ to $\EE_\rho$.

\begin{figure}[t]
\centering
\begin{minipage}[c]{0.50\textwidth}
\centering
\begin{tikzpicture}[scale=0.42,>=Stealth,line join=round]
    \def\tz{30}
    \def\dd{1}
    \begin{scope}
        \clip (-9.4,-6.0) rectangle (9.4,6.0);
        \begin{scope}
            \clip[rotate=-90]           (-60,-60) rectangle (\dd,60);
            \clip[rotate={2*\tz+90}]    (-60,-60) rectangle (\dd,60);
            \fill[BurntOrange!18]       (-60,-60) rectangle (60,60);
        \end{scope}
        \begin{scope}
            \clip[rotate=90]            (-60,-60) rectangle (\dd,60);
            \clip[rotate={2*\tz-90}]    (-60,-60) rectangle (\dd,60);
            \fill[BurntOrange!18]       (-60,-60) rectangle (60,60);
        \end{scope}
        \begin{scope}
            \clip[rotate=\tz]           (-60,-60) rectangle (\dd,60);
            \clip[rotate={\tz+180}]     (-60,-60) rectangle (\dd,60);
            \fill[RoyalBlue!16]         (-60,-60) rectangle (60,60);
        \end{scope}
        \begin{scope}
            \clip[rotate=\tz]           (-60,-60) rectangle (\dd,60);
            \clip[rotate={\tz+180}]     (-60,-60) rectangle (\dd,60);
            \clip[rotate=-90]           (-60,-60) rectangle (\dd,60);
            \clip[rotate={2*\tz+90}]    (-60,-60) rectangle (\dd,60);
            \fill[RoyalBlue!50!BurntOrange!45] (-60,-60) rectangle (60,60);
        \end{scope}
        \begin{scope}
            \clip[rotate=\tz]           (-60,-60) rectangle (\dd,60);
            \clip[rotate={\tz+180}]     (-60,-60) rectangle (\dd,60);
            \clip[rotate=90]            (-60,-60) rectangle (\dd,60);
            \clip[rotate={2*\tz-90}]    (-60,-60) rectangle (\dd,60);
            \fill[RoyalBlue!50!BurntOrange!45] (-60,-60) rectangle (60,60);
        \end{scope}
        \draw[gray!50,dashed] (180:9) -- (0:9);
        \draw[gray!50,dashed] ({2*\tz+180}:9) -- ({2*\tz}:9);
        \draw[BurntOrange!85,thick,rotate=-90]        (\dd,-60) -- (\dd,60);
        \draw[BurntOrange!85,thick,rotate=90]         (\dd,-60) -- (\dd,60);
        \draw[BurntOrange!85,thick,rotate={2*\tz+90}] (\dd,-60) -- (\dd,60);
        \draw[BurntOrange!85,thick,rotate={2*\tz-90}] (\dd,-60) -- (\dd,60);
        \draw[RoyalBlue!75,thick,rotate=\tz]          (\dd,-60) -- (\dd,60);
        \draw[RoyalBlue!75,thick,rotate={\tz+180}]    (\dd,-60) -- (\dd,60);
    \end{scope}
    \draw[->,thick] (0,0) -- (0:2.2)         node[below=-1pt]      {$e_-$};
    \draw[->,thick] (0,0) -- ({2*\tz}:2.2)   node[above right=-3pt]{$e_+$};
    \draw[->,thick] (0,0) -- ({\tz+90}:2.2)  node[left=-1pt]       {$e_0^\perp$};
    \fill (0,0) circle (2.2pt);
    \node[font=\small] at ({\tz}:7.0)        {$\supp Z$};
    \node[font=\small] at ({\tz+180}:7.0)    {$\supp Z$};
    \node[font=\small,align=center] at ({\tz+90}:4.9) {$\left|x\cdot e_0\right| \le 1$};
\end{tikzpicture}
\end{minipage}\hfill
\begin{minipage}[c]{0.47\textwidth}
\caption{The function $Z =\left(\kappa_--\kappa_+\right)\sigma_3$ is supported within distance $1$ of the cone spanning directions between $e_-$ and $e_+$ (orange); $W$ is supported in the strip $\left|x\cdot e_0\right| \le 1$ (blue). The two regions meet in a bounded set (shaded).}
\label{fig:supp-Z}
\end{minipage}
\end{figure}
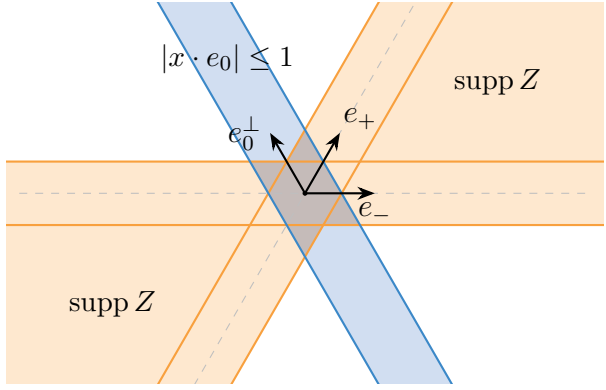

\textbf{6.} We now work on $W B_+(\lambda)\,Z\,B_-(\lambda)$. This term adds decay in the direction $e_0^\perp$, as well as in all directions not in the support of $Z$; these cover a basis of $\R^2$, see \Cref{fig:supp-Z}. We first observe that $Z$ is supported at distance at most $1$ from the set
\begin{equation}\label{eq:1f}
    \{ x : \  \text{ $x\cdot e_+^\perp$ and $x\cdot e_-^\perp$ have opposite signs} \} .
\end{equation}
In particular, on the support of $Z$, we have $x \cdot e_+^\perp \leq 1$ and $x \cdot e_-^\perp \geq -1$; or $x \cdot e_+^\perp \geq -1$ and $x \cdot e_-^\perp \leq 1$. The vector $e_0^\perp$ decomposes in the basis $(e_\pm,e_\pm^\perp)$ as 
\begin{equation}\label{eq:1j}
    e_0^\perp = \sin \theta_0 \, e_+ + \cos \te_0 \, e_+^\perp = -\sin \theta_0 \, e_- + \cos \te_0 \, e_-^\perp .
\end{equation}

Therefore, if $x \cdot e_+^\perp \leq 1$ and $x \cdot e_-^\perp \geq -1$, we have:
\begin{equation}\label{eq:cpt-e0perp}
\begin{aligned}
    x \cdot e_0^\perp &= \sin \te_0 \, x \cdot e_+ + \cos \te_0 \, x \cdot e_+^\perp \leq \sin \te_0 \, x \cdot e_+ + \cos \te_0 ;
    \\
    x \cdot e_0^\perp &= -\sin \te_0 \, x \cdot e_- + \cos \te_0 \, x \cdot e_-^\perp \geq -\sin \te_0 \, x \cdot e_- - \cos \te_0 ;
    \\
    \Rightarrow \qquad \left|x \cdot e_0^\perp\right| &\leq \left|\sin \te_0\right| \max\left(\left|x \cdot e_+\right| , \left|x \cdot e_-\right|\right) + \cos \te_0 .
\end{aligned}
\end{equation}
A similar argument shows that \eqref{eq:cpt-e0perp} also holds when $x \cdot e_+^\perp \geq -1$ and $x \cdot e_-^\perp \leq 1$. Therefore, for all $x \in \supp Z$ and $\rho \leq \tfrac{\gamma}{|\sin \te_0|}$,
\begin{align}
    \rho |x \cdot e_0^\perp| & \leq \rho |\sin \theta_0| \max (|x \cdot e_+| , |x \cdot e_-|) + \rho \cos \te_0
    \\
    & \leq \gamma \max (|x \cdot e_+| , |x \cdot e_-|) + \rho. \label{eq:1g}
\end{align}
In particular, this proves that 
\begin{equation}
    \big| Z(x) e^{\rho |x \cdot e_0^\perp|} \big| \leq \| Z \|_\infty e^{\gamma \max (|x \cdot e_+| , |x \cdot e_-|) + \rho} \leq \| Z \|_\infty e^\rho (e^{\gamma |x \cdot e_+|} + e^{\gamma |x \cdot e_-|}).
\end{equation}
This implies that $Z$, as a multiplication operator, maps
$e^{-\gamma |x \cdot e_+|} L^2 \cap e^{-\gamma |x \cdot e_-|} L^2$ to $e^{-\rho |x \cdot e_0^\perp|} L^2$.

Recall $\gamma < \gamma(\lambda) (1-|\sin \te_0|) < 2\gamma(\lambda)$ and assume that $0<\rho \leq \gamma(\lambda)(1-|\sin \te_0|)$; in particular $0<\rho \leq 2\gamma(\lambda)$.  
Bringing in \Cref{prop:agmon} -- and the fact that $W$ is supported in $|x \cdot e_0| \leq 1$ -- we obtain that the operators
\begin{equation}
\begin{aligned}
    B_-(\lambda) \ &: \ \EE_\gamma \ \longrightarrow \ e^{-\gamma \left|x \cdot e_+\right|} L^2 \cap e^{-\gamma \left|x \cdot e_-\right|} L^2 ,
    \\
    Z \ &: \ e^{-\gamma \left|x \cdot e_+\right|} L^2 \cap e^{-\gamma \left|x \cdot e_-\right|} L^2 \ \longrightarrow \ e^{-\rho \left|x \cdot e_0^\perp\right|} L^2 ,
    \\
    B_+(\lambda) \ &: \ e^{-\rho \left|x \cdot e_0^\perp\right|} L^2 \ \longrightarrow \ e^{-\rho \left|x \cdot e_0^\perp\right|} L^2 ,
    \\
    W \ &: \ e^{-\rho \left|x \cdot e_0^\perp\right|} L^2 \ \longrightarrow \ e^{-\rho \left|x \cdot e_0^\perp\right| - \rho \left|x \cdot e_0\right|} L^2 \ \subset \ \EE_\rho
\end{aligned}
\end{equation}
are bounded. Chaining these together, we obtain that $W B_+(\lambda)\,Z\,B_-(\lambda)$ maps $\EE_\gamma$ to $\EE_\rho$. 

\textbf{7.} Therefore, if $\rho$ satisfies
\begin{equation}\label{eq:cpt-rho}
    \gamma \ < \ \rho \ < \ \min\left( 1-\left|\sin\te_0\right|-\gamma \ , \ \ \gamma(\lambda)\left(1-\left|\sin\te_0\right|\right) \ , \ \ \dfrac{1-\left|\sin\te_0\right|}{2} \ , \ \ \dfrac{\gamma}{\left|\sin\te_0\right|} \right)
\end{equation}
-- with the convention $\gamma/0 = +\infty$ when $\te_0 = 0$ -- the terms $(1)$--$(5)$ from \eqref{eq:cpt-decomposition} are all bounded from $\EE_\gamma$ to $\EE_\rho$: the four upper bounds in \eqref{eq:cpt-rho} are the constraints produced in steps \textbf{3}, \textbf{4}, \textbf{5} and \textbf{6}, respectively. In particular,  if $\rho$ satisfies \eqref{eq:cpt-rho},
    $K_2(\lambda) : \EE_\gamma \rightarrow \EE_\rho$
 is bounded. 

It remains to check that there exist values $\rho$ satisfying \eqref{eq:cpt-rho}. We first note that
\begin{equation}
    \min\left( 1-\left|\sin\te_0\right|-\gamma \ , \ \ \gamma(\lambda)\left(1-\left|\sin\te_0\right|\right) \ , \ \ \dfrac{1-\left|\sin\te_0\right|}{2} \right) = \gamma(\lambda)\left(1-\left|\sin\te_0\right|\right)
\end{equation}
because $\gamma(\lambda) \leq 1/20$ and $\gamma < \tfrac{1-\left|\sin\te_0\right|}{2}$. By assumption, $\gamma(\lambda)\left(1-\left|\sin\te_0\right|\right) > \gamma$,
and $\tfrac{\gamma}{\left|\sin\te_0\right|} > \gamma$ (because $|\te_0| < \pi/2$), so the minimum in \eqref{eq:cpt-rho} is larger than $\gamma$ indeed. This completes the proof.
\end{proof}

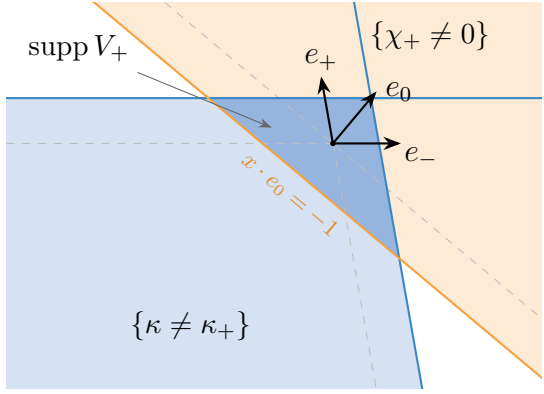
\begin{figure}[t]
\centering
\begin{minipage}[c]{0.54\textwidth}
\centering
\begin{tikzpicture}[scale=0.60,>=Stealth,line join=round]
    \def\tz{50}
    \def\dd{1}
    \begin{scope}
        \clip (-7.2,-5.4) rectangle (4.6,3.1);
        \begin{scope}
            \clip[rotate={\tz+180}]  (-60,-60) rectangle (\dd,60);
            \fill[BurntOrange!15]    (-60,-60) rectangle (60,60);
        \end{scope}
        \begin{scope}
            \clip[rotate=90]         (-60,-60) rectangle (\dd,60);
            \clip[rotate={2*\tz-90}] (-60,-60) rectangle (\dd,60);
            \fill[RoyalBlue!14]      (-60,-60) rectangle (60,60);
        \end{scope}
        \begin{scope}
            \clip[rotate=90]         (-60,-60) rectangle (\dd,60);
            \clip[rotate={2*\tz-90}] (-60,-60) rectangle (\dd,60);
            \clip[rotate={\tz+180}]  (-60,-60) rectangle (\dd,60);
            \fill[RoyalBlue!40]      (-60,-60) rectangle (60,60);
        \end{scope}
        \draw[gray!50,dashed] (0,0) -- (180:8);
        \draw[gray!50,dashed] (0,0) -- ({2*\tz+180}:8);
        \draw[gray!50,dashed] ({\tz+90}:8) -- ({\tz-90}:8);
        \draw[RoyalBlue!75,thick,rotate=90]          (\dd,-60) -- (\dd,60);
        \draw[RoyalBlue!75,thick,rotate={2*\tz-90}]  (\dd,-60) -- (\dd,60);
        \draw[BurntOrange!85,thick,rotate={\tz+180}] (\dd,-60) -- (\dd,60);
    \end{scope}
    \draw[->,thick] (0,0) -- (0:1.5)        node[below right=-3pt] {$e_-$};
    \draw[->,thick] (0,0) -- ({2*\tz}:1.5)  node[above=-2pt]       {$e_+$};
    \draw[->,thick] (0,0) -- ({\tz}:1.5)    node[right=-1pt]       {$e_0$};
    \fill (0,0) circle (1.7pt);
    \node[font=\small] at (232:5.1) {$\{\kappa \neq \kappa_+\}$};
    \node[font=\small] at (48:3.2)  {$\{\chi_+ \neq 0\}$};
    \node[font=\small] at (-5.6,2.1) (lab) {$\supp V_+$};
    \draw[->,black!60] (lab) -- (-1.35,0.35);
    \node[BurntOrange!90!black,font=\scriptsize,rotate={\tz-90},anchor=north,
          inner sep=2pt] at ({\tz+180}:{\dd+0.15}) {$x\cdot e_0 = -1$};
\end{tikzpicture}
\end{minipage}\hfill
\begin{minipage}[c]{0.46\textwidth}
\caption{The support of $V_+ = \chi_+\left(\kappa-\kappa_+\right)\sigma_3$. }
\label{fig:supp-V}
\end{minipage}
\end{figure}

\begin{lemma}\label{lem:cpt-fredholm} Under the assumptions of \Cref{lem:cpt-gain}, the operators $K_1(\lambda), K_2(\lambda) : \EE_\gamma \rightarrow \EE_\gamma$ are compact.
\end{lemma}

\begin{proof} \textbf{1.} The function $V_+$ is supported in
\begin{equation}
    \{ \kappa \neq \kappa_+ \} \cap \{ \chi_+ \neq 0 \}.
\end{equation}
These two sets lie within fixed distance of $\{ \pi < \arg x < \pi+\te_+ \}$ and $\{ x \cdot e_0 > 0 \}$, respectively; see \Cref{fig:supp-V}. So $V_+$ is a bounded, compactly supported function. Let $\varphi \in C^\infty_0(\R^2)$ be equal to $1$ on a neighborhood of $\supp V_+$, so that $V_+ R_+(\lambda) = V_+ \cdot \varphi R_+(\lambda)$. For $\lambda \in \tOmega$, the operator $\varphi R_+(\lambda)$ maps $\EE_\gamma$ boundedly to $H^1$, with values supported in the fixed compact set $\supp\varphi$; by the Rellich--Kondrachov theorem, it is therefore compact from $\EE_\gamma$ to $L^2$. Since $V_+$ is bounded and compactly supported, the multiplication operator $V_+ : L^2 \rightarrow \EE_\gamma$ is bounded. Being the composition of a compact and a bounded operator, $V_+R_+(\lambda)$ is thus compact as an operator on $\EE_\gamma$. The same holds for  $V_- R_-(\lambda)$. It follows that $K_1(\lambda)$ is compact on $\EE_\gamma$.

\textbf{2.} We now show that $K_2(\lambda) : \EE_\gamma \rightarrow \EE_\gamma$ is a limit (in operator norm) of compact operators on $\EE_\gamma$; hence it is itself compact. For $n \in \N$, let $\varphi_n$ be a smooth function supported in $B_{2n}(0)$, equal to $1$ on $B_n(0)$, such that $0 \leq\varphi_n \leq 1$. By \Cref{lem:cpt-gain}, the operator $K_2(\lambda)$ maps $\EE_\gamma$ to $\EE_\rho$ for some $\rho > \gamma$, and thus:
\begin{align}
    \left\| K_2(\lambda) - \varphi_n K_2(\lambda) \right\|_{\EE_\gamma \rightarrow \EE_\gamma} & \leq \left\| 1 - \varphi_n \right\|_{\EE_\rho \rightarrow \EE_\gamma} \left\| K_2(\lambda) \right\|_{\EE_\gamma \rightarrow \EE_\rho}
    \\
    & \leq C \left\| (1 - \varphi_n) e^{(\gamma-\rho) |x|}\right\|_\infty \leq C e^{-(\rho - \gamma) n}.
\end{align}
Hence, as operators on $\EE_\gamma$,
\begin{equation}
    K_2(\lambda) = \lim_{n \rightarrow \infty} \varphi_n K_2(\lambda).
\end{equation}
Because $R_\pm(\lambda)$ maps $\EE_\gamma$ to $H^1_\loc$ and $W$ is smooth with bounded derivatives, $K_2(\lambda)$ maps $\EE_\gamma$ to $H^1_\loc$, and hence $\varphi_n K_2(\lambda)$ maps $\EE_\gamma$ to $H^1_\comp$, in particular it is a compact operator on $\EE_\gamma$. This proves that $K_2(\lambda)$ is compact on $\EE_\gamma$, and this completes the proof. 
\end{proof}

\begin{proof}[Proof of \Cref{thm:mero}]
\textbf{1.}  We first claim that $\Id+K(\lambda)$ is invertible on $\EE_\gamma$ for every $\lambda$ satisfying \eqref{eq:cpt-parameters}, with $\Im\lambda>0$. By \Cref{lem:cpt-fredholm}, $K(\lambda)$ is compact, so it suffices to show that $\Id+K(\lambda)$  is injective.

Let $f \in \EE_\gamma$ with $\left(\Id+K(\lambda)\right)f = 0$. By \eqref{eq:parametrix-error}, we have
$(\Di-\lambda) Q(\lambda)f = 0$. Because $\Im \lambda > 0$, $\Di-\lambda$ is invertible on $L^2$ and hence $Q(\lambda) f = 0$. We write this identity as 
\begin{equation}
    0 = \chi_+u_++\chi_-u_-, \qquad u_\pm \defeq R_\pm(\lambda)f.
\end{equation}

Let $h \defeq u_+-u_- \in H^1$. Note that
\begin{equation}\label{eq:1s}
    \chi_- h = \chi_-u_+-\chi_-u_- =  \chi_-u_+ + \chi_+ u_+ = u_+.
\end{equation}
Therefore, from $u_\pm = R_\pm(\lambda)f$, we deduce that:
\begin{align}
      (\Di_--\lambda) u_- = f &= (\Di_+-\lambda) u_+  =  (\Di_--\lambda) u_+ + (\Di_+-\Di_-) u_+ \\ & = (\Di_--\lambda) u_+ + (\kappa_+-\kappa_-) \sigma_3 u_+ = (\Di_--\lambda) u_+ + (\kappa_+-\kappa_-) \chi_- \sigma_3 h.
\end{align}
In particular, subtracting $(\Di_--\lambda) u_-$ from both sides, we obtain:
\begin{equation}\label{eq:mero-selfadjoint}
    0 = \big(\Di_- - \lambda + \chi_-\left(\kappa_+-\kappa_-\right)\sigma_3 \big) h.
\end{equation}
The operator $\Di_- + \chi_-\left(\kappa_+-\kappa_-\right)\sigma_3$ is selfadjoint on $L^2$ with domain $H^1$. So \eqref{eq:mero-selfadjoint} is a spectral, selfadjoint problem with spectral parameter $\lambda$ with positive imaginary part. The only solution is $h=0$. This implies $u_+ =0$ by \eqref{eq:1s}, and so $f = (\Di_+-\lambda)u_+ = 0$. Therefore $\Id+K(\lambda)$  is injective for $\Im \lambda > 0$, and hence $\Id+K(\lambda)$ is invertible when in addition $\lambda$ satisfies \eqref{eq:cpt-parameters}, by the Fredholm alternative.

\textbf{2.} Fix $ 0 < \gamma < \left(1-\left|\sin\te_0\right|\right)\gamma(0)$ and set
\begin{equation}\label{eq:mero-domain}
    U_\gamma \defeq \left\{ \lambda \in \C \ : \ |\lambda| < E, \ \ \left|\Im\lambda\right| < \gamma < \left(1-\left|\sin\te_0\right|\right)\gamma(\lambda) \right\} .
\end{equation}
For $\lambda \in U_\gamma$:
\begin{itemize}
    \item $K(\lambda) : \EE_\gamma \rightarrow \EE_\gamma$ is analytic;
    \item $K(\lambda) : \EE_\gamma \rightarrow \EE_\gamma$ is compact on $\EE_\gamma$;
    \item $\Id+K(\lambda) : \EE_\gamma \rightarrow \EE_\gamma$ is invertible when $\Im \lambda > 0$.
\end{itemize}
Therefore, by analytic Fredholm theory \cite[Theorem C.8]{DZ}, $\left(\Id+K(\lambda)\right)^{-1}$ is a meromorphic family of bounded operators on $\EE_\gamma$, $\lambda \in U_\gamma$, whose Laurent coefficients at each pole are finite rank.

\textbf{3.} Now observe that
\begin{equation}\label{eq:mero-formula}
    (\Di-\lambda)^{-1} = Q(\lambda)\left(\Id+K(\lambda)\right)^{-1} \ : \ \EE_\gamma \longrightarrow L^2_\loc ,
\end{equation}
valid when $\Im \lambda > 0$ and $\lambda$ satisfies \eqref{eq:mero-domain}, meromorphically continues to $\lambda \in U_\gamma$. This provides a meromorphic continuation of $(\Di-\lambda)^{-1}$ from $\EE_\gamma$ to $L^2_\loc$ past $(-E,E)$.

In particular this meromorphic continuation maps $L^2_\comp$ to $L^2_\loc$; we denote this family by $R(\lambda)$ below. A priori $R(\lambda)$ depends on $\gamma$, but since $L^2_\comp$ is dense in $\EE_\gamma$, and they all continue the same operator for $\Im \lambda > 0$, the uniqueness of the meromorphic continuation shows that $R(\lambda)$ is $\gamma$-independent. In particular, we obtain a meromorphic continuation of $(\Di-\lambda)^{-1}$ from $L^2_\comp$ to $L^2_\loc$, for 
\begin{equation}
    \lambda \in \bigcup_{0 < \gamma < (1-|\sin \te_0|) \gamma(0)}U_\gamma = \tOmega.
\end{equation}

\textbf{4.} It remains to upgrade $L^2_\loc$ to $H^1_\loc$. Away from the poles, $u = R(\lambda)f$ satisfies $\left(\Di-\lambda\right)u = f$ in the sense of distributions: this holds for $\Im\lambda>0$ and continues meromorphically to $\lambda \in \tOmega$. Hence $\left(\Di-\lambda\right)u \in L^2_\loc$, and elliptic regularity for the first-order elliptic operator $\Di$ gives $u \in H^1_\loc$. By the closed graph theorem, $R(\lambda)$ maps $L^2_\comp$, respectively $\EE_\gamma$, to $H^1_\loc$. The proof is complete.
\end{proof}

\section{Scattering states}\label{sec:scattering}

\subsection{The scattering problem}\label{subsec:scattering-problem}

We are interested in the scattering properties of $\Di$. Throughout this section $\lambda \in (-E,E)$ is a fixed energy such that $R(\lambda)$ is well-defined on $\EE_\gamma$ for some $0 < \gamma < \cos\te_0$, and we write $\psi_{\pm,\lambda}$ for the edge states \eqref{eq:psi-lambda-e} of $\Di_\pm$ at energy $\lambda$:
\begin{equation}\label{eq:psi-pm}
    \psi_{\pm,\lambda}(x) \defeq e^{i\lambda\,x\cdot e_\pm}\,\phi_\pm\left(x\cdot e_\pm^\perp\right), \qquad\text{so that}\qquad \left(\Di_\pm - \lambda\right)\psi_{\pm,\lambda} = 0 .
\end{equation}
We are interested in solutions of
\begin{equation}\label{eq:scattering-eq}
    \left(\Di - \lambda\right)u = 0, \qquad u(x) \sim \psi_{-,\lambda}(x) \ \text{ as } \ x_1 \rightarrow -\infty,
\end{equation}
that are outgoing along $e_+$. This means that $u$ takes the form
\begin{equation}\label{eq:scattering-state}
    u = \rho\,\psi_{-,\lambda} + R(\lambda)f,
\end{equation}
for some $f$ decaying exponentially and $\rho$ such that
\begin{equation}\label{eq:cutoff-rho}
    \rho \in C^\infty(\R), \qquad \rho(x_1) = 1 \ \text{ for } x_1 \le -R-2, \qquad \rho(x_1) = 0 \ \text{ for } x_1 \ge -R-1,
\end{equation}
with $R$ as in \Cref{def:bent}.
Such functions $u$ are called distorted plane waves at energy $\lambda$. They describe the scattering properties of the system. Indeed, superposing them over $\lambda$ produces wavepackets initially propagating rightwards, along $\R e_-$. The asymptotics of $u$ allow us to study transmission and reflection by the corner.

The next result  identifies the function $f$ in \eqref{eq:scattering-state}.

\begin{lemma}\label{lem:source}
With $\Di$ as in \eqref{eq:dirac-bent},
\begin{equation}\label{eq:source}
    \left(\Di - \lambda\right)\left(\rho\,\psi_{-,\lambda}\right) = \left(D_1\rho\right)\psi_{-,\lambda} + \rho\left(\kappa-\kappa_-\right)\sigma_3\,\psi_{-,\lambda} \defeq -f.
\end{equation}
Moreover, $f \in \EE_\gamma$ for all $\gamma < \cos \te_0$. 
\end{lemma}

\begin{proof}
The operators $\Di$ and $\Di_-$ differ only by their mass terms: $\Di = \Di_- + \left(\kappa-\kappa_-\right)\sigma_3$. Since $\rho$ depends only on $x_1$ and $\left(\Di_- - \lambda\right)\psi_{-,\lambda} = 0$ by \eqref{eq:psi-pm}, we obtain
\begin{equation}
\begin{aligned}
    \left(\Di - \lambda\right)\left(\rho\,\psi_{-,\lambda}\right)
     = \left(D_1\rho\right)\sigma_1\psi_{-,\lambda} + \rho\left(\kappa-\kappa_-\right)\sigma_3\,\psi_{-,\lambda} .
\end{aligned}
\end{equation}
Because $\sigma_1\phi_- = \phi_-$, we have $\sigma_1\psi_{-,\lambda} = \psi_{-,\lambda}$. This proves \eqref{eq:source}.

We split $f = f_1 + f_2$,
\begin{equation}
    f_1 = \left(D_1\rho\right)\psi_{-,\lambda} , \qquad f_2 = \rho\left(\kappa-\kappa_-\right)\sigma_3\,\psi_{-,\lambda}\add{,}
\end{equation}
and we estimate $f_1, f_2$ separately. Because $D_1\rho$ is supported in a vertical strip and $\psi_{-,\lambda}$ decays exponentially like $e^{-|x_2|}$, $f_1 \in \EE_\gamma$ for all $\gamma < 1$, in particular for all $\gamma < \cos \te_0$. 

The function $f_2$ is supported in $x_1 \leq -R-1$, $x\cdot e_0 \ge 0$. In particular, for $x \in \supp f_2$, we have:
\begin{equations}
    x_2\sin\theta_0 \ge -x_1\cos\theta_0 = |x_1|\cos\theta_0,
\\
    \cos^2\theta_0\,|x|^2 = x_1^2\cos^2\theta_0 + x_2^2\cos^2\theta_0 \leq x_2^2\sin^2\theta_0 + x_2^2\cos^2\theta_0 = x_2^2. 
\end{equations}
Because $\psi_{-,\lambda}$ decays like $e^{-|x_2|}$, $f_2$ decays like $e^{-\cos \theta_0 |x|}$. In particular, $f_2 \in \EE_\gamma$ for every $\gamma < \cos \te_0$. This completes the proof. 
\end{proof}

In particular, if $R(\lambda)$ is well-defined on $\EE_\gamma$ for some $\gamma < \cos \te_0$, then  \eqref{eq:scattering-eq} has a unique solution $u$ that is outgoing along $e_+$ in the sense of \eqref{eq:scattering-state}. It is given by: 
\begin{equations}\label{eq:1c}
    u = \rho\,\psi_{-,\lambda} + R(\lambda)f, \qquad f \defeq -\left(D_1\rho\right)\psi_{-,\lambda} - \rho\left(\kappa-\kappa_-\right)\sigma_3\,\psi_{-,\lambda}.
\end{equations} 
We refer to $u$ as a \textit{distorted plane wave of $\Di$ at energy $\lambda$.}

\begin{remark}\label{rem:cutoff-independent}
The distorted plane wave $u$ at energy $\lambda$ does not depend on the choice of $\rho$. Indeed, let $\rho_1,\rho_2$ both satisfy \eqref{eq:cutoff-rho}, $f_j$ be associated to $\rho_j$ by \eqref{eq:source}, and $u_j$ be associated to $\rho_j, f_j$ by \eqref{eq:scattering-state} ($j=1,2$). We have:
\begin{align}
    u_1-u_2 & = (\rho_1-\rho_2)\psi_{-,\lambda} + R(\lambda)(f_1-f_2)
    \\
    & =(\rho_1-\rho_2)\psi_{-,\lambda} - R(\lambda)(\Di-\lambda)(\rho_1-\rho_2)\psi_{-,\lambda} = 0,
\end{align}
where we used that $v = (\rho_1-\rho_2)\psi_{-,\lambda} \in \EE_\gamma$ for all $\gamma < 1$ and $R(\lambda)(\Di-\lambda) v = v$ for exponentially decaying $v$.
\end{remark}

Our goal is now to prove the main result of this work, \Cref{thm:intro}, recast here in a more detailed form:

\begin{theorem}\label{thm:scattering} There exists a locally finite set $S \subset (-E,E)$ with the following property. For every $\lambda \in (-E,E) \setminus S$:
    \begin{itemize}
        \item[(a)] There exists $0< \gamma < \cos \te_0$ such that $R(\lambda): \EE_\gamma \rightarrow L^2_\loc$ is well-defined and $\Id + K(\lambda)$ is a linear invertible map on $\EE_\gamma$;
        \item[(b)] There exists $\nu >0$ such that the function $u$ given by \eqref{eq:1c} satisfies:
        \item[] \begin{equations}\label{eq:scattering-asymptotics}
    u \; = \; \rho\,\psi_{-,\lambda} \; + \; i \widehat{a_+}(\lambda)\cdot \chi_+\,\psi_{+,\lambda} \; + \; \EE_\nu ,
    \qquad g \defeq \big(\Id+K(\lambda)\big)^{-1} f,
    \\
    a_+(s)=\int_\R\ove{\phi_+(\tau)}^\top g(se_++\tau e_+^\perp)d\tau,  \qquad \widehat{a_+}(\lambda) \defeq \int_\R e^{-i\lambda s}\,a_+(s)\,ds.
\end{equations}
Moreover, $\left|\widehat{a_+}(\lambda)\right| = 1$.
    \end{itemize}
\end{theorem}

\begin{remark}
In the proof, $S$ emerges as the set of characteristic values of $K(\lambda)$, i.e. as the values of $\lambda$ such that $\Id +K(\lambda)$ is not invertible. Because of the identity
\begin{equation}
R(\lambda) =    Q(\lambda) \big( \Id + K(\lambda) \big)^{-1},
\end{equation}
one may be tempted to describe $S$ as the set of resonances of $\Di$ in $(-E,E)$ -- hence of eigenvalues of $\Di$ in $(-E,E)$. It is in fact more subtle, because $Q(\lambda)$ may cancel a pole of $\big( \Id + K(\lambda) \big)^{-1}$. In other words, characteristic values of $K(\lambda)$ in $(-E,E)$ form a larger set than eigenvalues of $\Di$.
\end{remark}

\subsection{A current identity}\label{subsec:transmission}

We start with a useful identity:

\begin{lemma}\label{lem:current}
If $v, w \in H^1_\loc$ then in the sense of distributions:
\begin{equation}\label{eq:current-identity}
    \ove{v}^\top \left(\Di w\right) \ - \ \ove{\left(\Di v\right)}^\top w = -i \cdot
    \operatorname{div} J(v,w), \qquad J(v,w) \ \text{ as in } \eqref{eq:current} .
\end{equation}
In particular, if $u \in H^1_\loc$ satisfies $\Di u = \lambda u$ for some $\lambda \in \R$, then
\begin{equation}\label{eq:current-eigen}
    \operatorname{div} J(u,u) = 0 .
\end{equation}
\end{lemma}

\begin{proof}
Write $\langle a,b\rangle \defeq \ove a^\top b$ for the pairing of $\C^2$, conjugate linear in the first slot. We have
\begin{equation} \label{eq:current-second-half}
    \sum_j \sigma_j\partial_j w = i\left(\Di w - \kappa\sigma_3 w\right), \qquad
    \sum_{j=1,2} \left\langle v, \sigma_j\partial_j w\right\rangle
    = i\left\langle v, \Di w\right\rangle - i\kappa \left\langle v,\sigma_3 w\right\rangle .
\end{equation}
The matrices $\sigma_j$ being Hermitian, the same computation applied to $v$ gives
\begin{equation}\label{eq:current-first-half}
    \sum_{j=1,2}\left\langle \partial_j v, \sigma_j w\right\rangle
    = \left\langle \sum_{j=1,2}\sigma_j\partial_j v, w\right\rangle
    = \left\langle i\left(\Di v-\kappa\sigma_3 v\right),w\right\rangle
    = -i\left\langle \Di v,w\right\rangle + i\kappa\left\langle v,\sigma_3 w\right\rangle.
\end{equation}
Adding \eqref{eq:current-second-half} and \eqref{eq:current-first-half} gives \eqref{eq:current-identity}. Taking $v=w=u$, the components of $J(u,u)$ are real because $\sigma_1,\sigma_2$ are Hermitian, and \eqref{eq:current-identity} becomes 
\begin{equation}
\operatorname{div}J(u,u) = i\left(\lambda - \ove\lambda\right)|u|^2 = -2\,\Im\lambda  \cdot|u|^2  = 0.
\end{equation}
This completes the proof.
\end{proof}

\subsection{Far field asymptotics}

\begin{lemma}\label{lem:outgoing-far-field}
    Assume given
    \begin{equation}
        \lambda \in (-E,E), \qquad 0<\gamma \leq \gamma(\lambda), \qquad g \in \EE_\gamma,
    \end{equation}
and let $u = Q(\lambda) g$. Then, with $a_+$ given in \eqref{eq:scattering-asymptotics},
    \begin{equations}\label{eq:outgoing-far-field}
        u(x) = i\,\widehat{a_+}(\lambda) \cdot \chi_+(x)\,e^{i\lambda\,x\cdot e_+}\,\phi_+\left(x\cdot e_+^\perp\right) + \EE_\nu, \qquad \nu \defeq \dfrac{\gamma}{2} \sin\left( \dfrac{\pi-|\te_+|}{8} \right).
    \end{equations}
\end{lemma}

\begin{proof} \textbf{1.}
By \eqref{eq:parametrix} and \eqref{eq:2z}, $u = v_++v_-+\EE_{\gamma/2}$, where
\begin{equation}\label{eq:v-pm-out}
    v_\pm(x) \defeq i\,\chi_\pm(x)\,e^{i\lambda\,x\cdot e_\pm}\left(\int_{-\infty}^{x\cdot e_\pm}e^{-i\lambda s}\,a_\pm(s)\,ds\right)\phi_\pm\left(x\cdot e_\pm^\perp\right) .
\end{equation}
By \eqref{eq:edge-data} and the Cauchy--Schwarz inequality, we have
\begin{equation}\label{eq:tail-bounds}
    \left|\int_{-\infty}^{t}e^{-i\lambda s}\,a_\pm(s)\,ds\right| \, \le \, C e^{\gamma\min(t,0)} ,
    \qquad
    \left|\int_{t}^{+\infty}e^{-i\lambda s}\,a_\pm(s)\,ds\right| \, \le \, C e^{-\gamma\max(t,0)};
\end{equation}
in particular $\widehat{a_+}(\lambda)$ is well defined. 

Let     $\epsi = \tfrac{\pi-|\te_+|}{8}$,
so that $\nu \leq \tfrac{\gamma}{2} \sin \epsi$, and define the following sets:
\begin{equations}\label{eq:covering}
    S_1 \defeq \left(-\pi+\epsi,-\epsi\right) \cup \left(\epsi,\pi-\epsi\right), \qquad
    S_2 \defeq \left(-\pi,-\tfrac\pi2-\epsi\right) \cup \left(\tfrac\pi2+\epsi,\pi\right] ,
    \\
    S_3 \defeq \left(\te_0-\tfrac\pi2+\epsi,\ \te_0+\tfrac\pi2-\epsi\right).
\end{equations}

We have the following:
\begin{equations}\label{eq:covering-bounds}
    \arg x \in S_1 \Rightarrow \left|x\cdot e_-^\perp\right| = |x|\left|\sin\arg x\right| \ge |x| \sin\epsi ,
    \\
    \arg x \in S_2 \Rightarrow x\cdot e_- = |x| \cos\arg x \le -|x| \sin\epsi ,
    \\
    \arg x \in S_3 \Rightarrow x\cdot e_0 = |x| \cos\left(\arg x - \te_0\right) \ge |x| \sin\epsi .
\end{equations}
These three sets cover $\left(-\pi,\pi\right]$. Indeed $S_1 \cup S_2$ covers $\left(-\pi,\pi\right] \setminus \left[-\epsi,\epsi\right]$, and $\left[-\epsi,\epsi\right] \subset S_3$ because $\te_0-\tfrac\pi2+\epsi<-\epsi$ and $\epsi<\te_0+\tfrac\pi2-\epsi$. In particular, at least one of the implications \eqref{eq:covering-bounds} applies to every $x \neq 0$.

\begin{figure}[t]
\centering
\begin{minipage}[c]{0.40\textwidth}
    \centering
    \includegraphics[width=\linewidth]{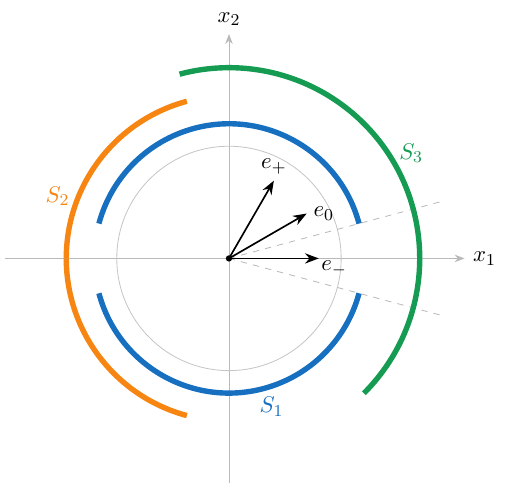}
\end{minipage}\hfill
\begin{minipage}[c]{0.55\textwidth}
\caption{The sets $S_1, S_2$ and $S_3$.}
\label{fig:covering}
\end{minipage}
\end{figure}

\textbf{2.} The function $\chi_-$ vanishes on $\left\{x\cdot e_0\ge1\right\}$, so that $\chi_-(x) \le C e^{\min\left(-x\cdot e_0,\,0\right)}$; together with the first bound of \eqref{eq:tail-bounds} and with $\left|\phi_-(\tau)\right|\le Ce^{-|\tau|}$, this gives
\begin{equation}\label{eq:v-minus-bound}
    \left|v_-(x)\right| \; \le \; C\exp\left(\gamma\min\left(-x\cdot e_0,0\right) + \gamma\min\left(x\cdot e_-,0\right) - \gamma\left|x\cdot e_-^\perp\right|\right) .
\end{equation}
The three terms of that exponent are nonpositive, and by \eqref{eq:covering-bounds} at least one of them is less than $-\gamma\sin\epsi\,|x|$. Hence $\left|v_-(x)\right| \le Ce^{-\gamma\sin\epsi\,|x|}$, and $v_- \in \EE_\nu$ since $\nu<\gamma\sin\epsi$.

\textbf{3.} We write
\begin{equations}\label{eq:v-plus-split}
    v_+(x) = i\,\widehat{a_+}(\lambda)\cdot\chi_+(x)\,e^{i\lambda\,x\cdot e_+}\phi_+\left(x\cdot e_+^\perp\right) \; - \; T_+(x),
    \\
    T_+(x) \defeq i\,\chi_+\,e^{i\lambda\,x\cdot e_+}\left(\int_{x\cdot e_+}^{+\infty}e^{-i\lambda s}a_+(s)\,ds\right)\phi_+\left(x\cdot e_+^\perp\right) .
\end{equations}
The second bound of \eqref{eq:tail-bounds} gives
\begin{equation}\label{eq:T-bound}
    \left|T_+(x)\right| \; \le \; C\,\chi_+(x)\exp\left(-\gamma\max\left(x\cdot e_+,0\right) - \gamma\left|x\cdot e_+^\perp\right|\right) .
\end{equation}
Set $\alpha \defeq \arg x - \te_+$ and recall that $\te_+ = 2\te_0$. We have the following:
\begin{equations}\label{eq:covering-bounds-plus}
    \alpha \in S_1 \Rightarrow \left|x\cdot e_+^\perp\right| = |x|\left|\sin\alpha\right| \ge |x|\sin\epsi ,
    \\
    \left|\alpha\right| \le \epsi \Rightarrow x\cdot e_+ = |x|\cos\alpha \ge |x|\cos\epsi ,
    \\
    \left|\alpha\right| \ge \pi-\epsi \Rightarrow x\cdot e_0 = |x|\cos\left(\alpha+\te_0\right) \le -\cos\left(\left|\te_0\right|+\epsi\right)|x| .
\end{equations}
Since $S_1$, $\left[-\epsi,\epsi\right]$ and $\left\{\left|\alpha\right| \ge \pi-\epsi\right\}$ cover all possible arguments, at least one of the implications \eqref{eq:covering-bounds-plus} applies to every $x \neq 0$. In the third case $x\cdot e_0 \le -1$ as soon as $|x| \ge \cos\left(\left|\te_0\right|+\epsi\right)^{-1}$, so that $\chi_+(x) = 0$ there. Since $\epsi<\tfrac\pi4$, each region therefore contributes $O\!\left(e^{-\gamma\sin\epsi|x|}\right)$ to \eqref{eq:T-bound}; hence $T_+ \in \EE_\nu$, again because $\nu<\gamma\sin\epsi$. Going back to 
\eqref{eq:v-plus-split}, we conclude that 
\begin{equations}
    v_+(x) = i\,\widehat{a_+}(\lambda)\cdot\chi_+(x)\,e^{i\lambda\,x\cdot e_+}\phi_+\left(x\cdot e_+^\perp\right) + \EE_\nu.
\end{equations}
This completes the proof.
\end{proof}

\begin{proof}[Proof of \Cref{thm:scattering}] \textbf{1.} Fix $\delta > 0$. It suffices to show that the conclusions (a) and (b) of \Cref{thm:scattering} hold for all but finitely many values of $\lambda$ in $[-E+\delta,E-\delta]$.

Define
\begin{equation}\label{eq:1d}
    \gamma \defeq \frac{1}{2}\min\left(\cos\te_0, \ \left(1-\left|\sin\te_0\right|\right)\gamma\left(E-\frac{\delta}{2}\right)\right).
\end{equation}
Note that $0 < \gamma < \left(1-\left|\sin\te_0\right|\right)\gamma(0)$. In particular, the conclusion of Step 2 in the proof of \Cref{thm:mero} applies. Specifically, for
\begin{equation}
    \lambda \in U_\gamma  = \left\{ \lambda \in \C \ : \ |\lambda| < E, \ \ \left|\Im\lambda\right| < \gamma < \left(1-\left|\sin\te_0\right|\right)\gamma(\lambda) \right\},
\end{equation}
the operators 
$\left(\Id+K(\lambda)\right)^{-1} : \EE_\gamma \rightarrow \EE_\gamma$ depend meromorphically on $\lambda$. 

Moreover, $[-E+\delta,E-\delta]$ is a compact subset of $U_\gamma$: indeed, $U_\gamma$ is an open set which contains $(-E+\delta/2,E-\delta/2)$. Therefore, $\left(\Id+K(\lambda)\right)^{-1}$ has at most finitely many poles in $[-E+\delta,E-\delta]$. In other words: for all but finitely many values of $\lambda $ in $[-E+\delta,E-\delta]$,
\begin{itemize}
    \item[(i)] $\Id + K(\lambda) : \EE_\gamma \rightarrow \EE_\gamma$ is invertible; 
    \item[(ii)] and by \eqref{eq:mero-formula}, $R(\lambda) : \EE_\gamma \rightarrow L^2_\loc$ is well-defined;
    \item[(iii)] hence conclusion (a) in \Cref{thm:scattering} holds.
\end{itemize}
In the rest of the proof, we assume that $\lambda \in [-E+\delta,E-\delta]$ is away from this finite subset of values, and we prove that conclusion (b) in \Cref{thm:scattering} holds.

    \textbf{2.} Let $u_-(x) = \rho(x_1) \psi_{-,\lambda}(x)$. Because of \eqref{eq:scattering-state} and of \eqref{eq:mero-formula}, $u - u_- = R(\lambda)f$ where $f \in \EE_\gamma$ is given by \eqref{eq:source}. Because $\Id + K(\lambda)$ is invertible by our choices of $\gamma$ and $\lambda$, we have $u - u_-= Q(\lambda)g$ with
    \begin{equation}
 g \defeq \big(\Id+K(\lambda)\big)^{-1}f \in \EE_\gamma.
    \end{equation}
     In particular, $u-u_-$ satisfies the asymptotics \eqref{eq:outgoing-far-field}. By \Cref{lem:outgoing-far-field} -- note that $\gamma \leq \gamma(\lambda)$ by \eqref{eq:1d} -- we have:
    \begin{equation}
        u = u_-+u_+ + \EE_\nu, \qquad u_+(x) \defeq i\,\widehat{a_+}(\lambda)\cdot \chi_+(x)\,e^{i\lambda\,x\cdot e_+}\,\phi_+\left(x\cdot e_+^\perp\right) ,
        \qquad \nu = \frac{\gamma}{2}\sin \left( \dfrac{\pi-|\te_+|}{8} \right).
    \end{equation}
This proves \eqref{eq:scattering-asymptotics}. All that remains to show is the identity $\left|\widehat{a_+}(\lambda)\right| = 1$.

\textbf{3.} By the Cauchy--Schwarz inequality, $\EE_\nu \subset L^1$, $\EE_\nu \cdot \EE_\nu \subset L^1$, and $L^\infty \cdot \EE_\nu \subset L^1$. Therefore, expanding $J(u,u)$ according to $u = u_-+u_++\EE_\nu$ and using $u_\pm \in L^\infty$ produces:
\begin{equation}\label{eq:1v}
    J(u,u) = J(u_-,u_-) + 2\Re J(u_-,u_+) + J(u_+,u_+) +L^1.
\end{equation}

\textbf{4.} Since $u$ solves $\left(\Di-\lambda\right)u = 0$ with $\lambda \in (-E,E)$ real, \Cref{lem:current} gives $\operatorname{div}J(u,u) = 0$, and hence by the divergence theorem, we have for each $r > 0$:
\begin{equation}\label{eq:1a}
   0 = \int_{B_r(0) } \operatorname{div} J(u,u) dx = \int_{\partial B_r(0) } J(u,u) \cdot n \ d\ell.
\end{equation}
We now exploit this identity by evaluating each boundary integral corresponding to the decomposition \eqref{eq:1v} of $J(u,u)$. 

After adapting \Cref{lem:edge-flux} to the swapped asymptotics $\rho(x_1) = 1$ near $-\infty$ and $0$ near $+\infty$, the divergence theorem applied to $J(u_-,u_-)(x) = \rho(x_1)^2\left|\phi_-(x_2)\right|^2\,e_1$ gives
\begin{equation}\label{eq:B-v-v-flux}
    \int_{\partial B_r(0)} J(u_-,u_-)\cdot n \ d\ell
    \; = \; \int_{B_r(0)}\left(\rho^2\right)'(x_1)\left|\phi_-(x_2)\right|^2 dx
    \; \xrightarrow[\ r \rightarrow +\infty\ ]{} \; \int_\R \left(\rho^2\right)'(x_1)\,dx_1 \; = \; -1 .
\end{equation}
Likewise,
\Cref{lem:edge-flux-tilted} shows that
\begin{equation}\label{eq:Juplus}
    \int_{\partial B_r(0)} J(u_+,u_+)\cdot n \ d\ell
    \; \xrightarrow[\ r \rightarrow +\infty\ ]{} \; \big| \widehat{a_+}(\lambda) \big|^2.
\end{equation}

\textbf{5.} We now study the cross term $J(u_-,u_+)$. The function $u_-$ has support in $x_1 \leq 0$ and the function $u_+$ has support in $x \cdot e_0 \geq -1$. Therefore, on the support of $J(u_-,u_+)$, we have
\begin{equation}
    |x_1|\cos\te_0 = -x_1\cos\te_0 = -x \cdot e_0 + x_2\sin\te_0 \leq 1 + |x_2|.
\end{equation}
Since moreover $u_-$ decays like $e^{-|x_2|}$, we deduce that 
\begin{equation}
    \big|J(u_-,u_+)(x)\big| \leq C e^{-|x_2|} \leq C \exp\left( {-\tfrac{|x_2|}{2} - \tfrac{|x_1|\cos\te_0}{2}} \right).
\end{equation}
In particular, $J(u_-,u_+) \in L^1$. Going back to \eqref{eq:1v}, we obtain:
\begin{equation}\label{eq:1e-split}
    J(u,u) = J(u_+,u_+) + J(u_-,u_-) + h, \qquad h \in L^1.
\end{equation}

\textbf{6.} Note that
\begin{equation}
   +\infty > \int_{\R^2} |h| = \int_{\R^+} \left(\int_{|x| = r} |h| d\ell\right) dr.
\end{equation}
Therefore there exists a sequence $r_n \rightarrow + \infty$ such that
\begin{equation}\label{eq:h-flux}
    \int_{\partial B_{r_n}(0)} \left|h\right| \, d\ell \; \xrightarrow[\ n\to+\infty\ ]{} \; 0 .
\end{equation}

We deduce that as $n \rightarrow \infty$,
\begin{align}
   0 = \int_{\partial B_{r_n}(0)} J(u,u)\cdot n \ d\ell & = \int_{\partial B_{r_n}(0)} J(u_-,u_-)\cdot n \ d\ell + \int_{\partial B_{r_n}(0)} J(u_+,u_+)\cdot n \ d\ell   + o(1)
   \\
   & 
   \rightarrow -1+\big|\widehat{a_+}(\lambda)\big|^2.
\end{align}
In the first identity, we used \eqref{eq:1a}; in the second, we used the decomposition \eqref{eq:1e-split} as well as the asymptotics \eqref{eq:h-flux} for $h$; and in the third, we used the asymptotics \eqref{eq:B-v-v-flux} and \eqref{eq:Juplus} for $J(u_-,u_-)$ and $J(u_+,u_+)$, respectively. We therefore conclude that $| \widehat{a_+}(\lambda) |^2=1$. This completes the proof.
\end{proof}

\section{Numerics}\label{sec:numerics2}

 We illustrate the results of \Cref{sec:bent} and \Cref{sec:scattering} dynamically, by solving numerically the Dirac evolution $\left(D_t + \Di\right)u = 0$ for four domain walls.

\subsection{The domain walls}\label{subsec:numerics-walls2}

All four functions $\kappa$ are bent domain walls in the sense of \Cref{def:bent}. They are step functions, taking the values $\pm1$:
\begin{equation}\label{eq:numerics-walls2}
\begin{aligned}
    \text{(a)} \quad & \kappa = 1 - 2\cdot\1_{\{x_1<0,\ x_2>0\}} , && \text{a right-angle corner};
    \\
    \text{(b)} \quad & \kappa = 1 - 2\cdot\1_{\{x_2>0,\ x_1+2x_2<0\}} , && \text{a hairpin};
    \\
    \text{(c)} \quad & \kappa = 1 - 2\cdot\1_{\{x_2>0,\ x_1<0\} \cup \{x_2>3,\ x_1>0\}} , && \text{a staircase};
    \\
    \text{(d)} \quad & \kappa = 1 - 2\cdot\1_{\{x_2>0\}} + 2\cdot\1_{\{0<x_1<1,\ 0<x_2<3\}} , && \text{a tooth}.
\end{aligned}
\end{equation}
The interfaces, that is the boundaries between $\{\kappa = 1\}$ and $\{\kappa = -1\}$, are drawn in the first row of \Cref{fig:numerics-evolution2}. All four contain the horizontal branch $\{x_2=0,\ x_1 \le -1\}$. They all satisfy \Cref{def:bent}. Indeed, (a), (b) and (d) are examples of the form \eqref{eq:bent-examples}; as for (c), we have 
\begin{equation}
    1 - 2\cdot\1_{\{x_2>0,\ x_1<0\} \cup \{x_2>3,\ x_1>0\}} = \begin{cases} \1_{x_2 <0} - \1_{x_2>0} \quad \text{ when $x_1 < 0$} \\  
        \1_{x_2 < 3} - \1_{x_2>3} \quad \text{ when $x_1 > 0$}
    \end{cases},
\end{equation}
so it corresponds to the junction of two horizontal (yet distinct) domain walls.

Along the incoming branch, that is for $x_1 \le -1$, all four walls agree with a horizontal domain wall $\kappa_*(x) = k_*(x_2)$ in the sense of \Cref{def:domain-wall}. We take for $u_0$ the corresponding edge state \eqref{eq:phi0} modulated by a Gaussian centered at $x_1 = -10$:
\begin{equation}\label{eq:u0-numerics2}
    u_0(x) \defeq \frac{1}{\sqrt{2\pi}}\,e^{-\frac{(x_1+10)^2}{2}}\,\phi_*(x_2) .
\end{equation}
It is the same for the four walls \eqref{eq:numerics-walls2}, since all four agree with $\kappa_*$ near $x_1 = -10$. The function $u_0$ is a superposition of the edge states \eqref{eq:psi-lambda}:
\begin{equation}\label{eq:u0-spectral2}
    u_0(x) = \frac{1}{2\pi}\int_\R e^{-\frac{\lambda^2}{2}}\,e^{10\,i\lambda}\,\psi_{\lambda,*}(x)\,d\lambda .
\end{equation}
It is not perfectly spectrally supported in $(-1,1)$: the spectral weight along the edge branch is proportional to $e^{-\lambda^2}$, so the fraction of $u_0$ carried by energies in the gap is $\operatorname{erf}(1) \simeq 84\%$. Around $16\%$ of $u_0$ is carried outside the gap.

\subsection{The scheme}\label{subsec:numerics-scheme2}

Space is confined to $[-24,24)^2$, with $512$ points in each direction. We perform the simulation using the splitting
\begin{equation}
    \Di = T + M, \qquad T \defeq D_1\sigma_1+D_2\sigma_2, \qquad M \defeq \kappa(x)\sigma_3.
\end{equation}
The flows of both $M$ and $T$ are available in closed form:
\begin{equation}\label{eq:split-flows2}
    e^{itM} = \begin{bmatrix} e^{it\kappa(x)} & 0 \\ 0 & e^{-it\kappa(x)} \end{bmatrix},
    \qquad
    e^{it\,\sigma\cdot\xi} = \cos\left(|\xi|t\right)\Id + i\,\frac{\sin\left(|\xi|t\right)}{|\xi|}\,\sigma\cdot\xi,
\end{equation}
where the second flow is expressed on the Fourier side.
To compute $e^{-it\Di}$, we iterate
\begin{equation}\label{eq:strang2}
    e^{-it\Di} \; = \; \left( e^{-i\frac{\delta}{2}M}\,e^{-i\delta T}\,e^{-i\frac{\delta}{2}M} \right)^{t/\delta} \; + \; O\left(\delta^2\right).
\end{equation}
Each factor is unitary, so the scheme is norm-preserving and unconditionally stable.

\subsection{Results}\label{subsec:numerics-results2}

\Cref{fig:numerics-evolution2} shows $\left|u(t,x)\right|$ at $t=6,12,18,24$ for the four walls \eqref{eq:numerics-walls2}. The packet travels along the incoming branch at speed one, reaches the defect at $t \simeq 10$, and in every case it follows the interface around the defect and leaves along the outgoing branch, again at speed one.

\begin{figure}[!ht]
\centering
\includegraphics[width=0.94\textwidth]{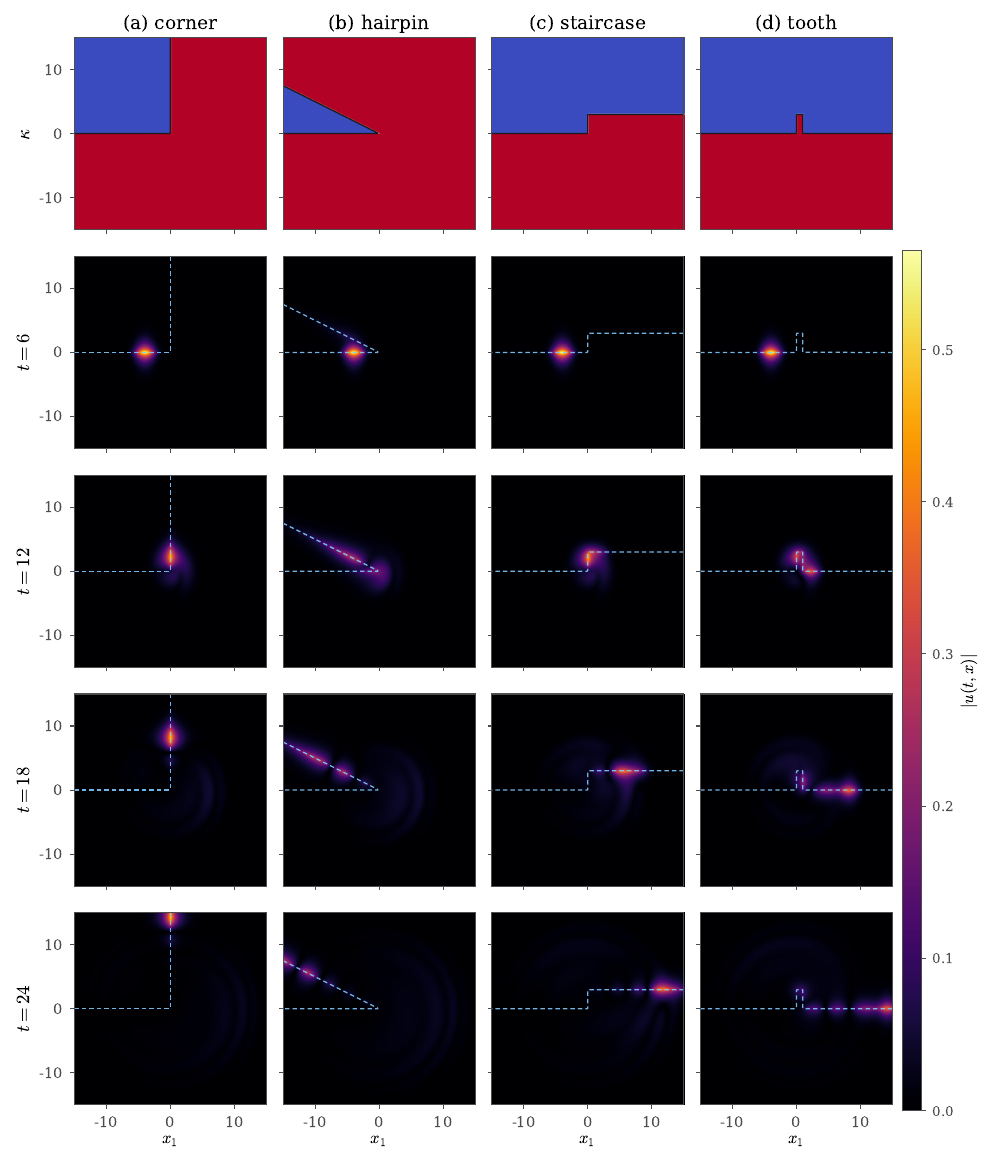}
\caption{The Dirac evolution $\left(D_t + \Di\right)u = 0$ for the four domain walls of \eqref{eq:numerics-walls2}. Columns: (a) right-angle corner, (b) hairpin, (c) staircase, (d) tooth. First row: the mass $\kappa$, from $-1$ (blue) to $+1$ (red), with the interface in black. Remaining rows: $\left|u(t,x)\right|$ at $t=6,12,18,24$, with the interface dashed. The packet reaches the defect at $t \simeq 10$ and follows the interface around it. A faint circular wave visible for $t \geq 12$ gives the bulk radiation as the solution passes the defect.}
\label{fig:numerics-evolution2}
\end{figure}

Two features are common to the four runs. First, nothing measurable comes back. The fraction of $\|u(t)\|^2_{L^2}$ that has returned to the incoming branch at $t=24$, measured on $\{x_1 \le -8,\ |x_2| \le 2.5\}$, is:
\begin{equation}
    \text{$3\cdot10^{-4} \ $ for (a), $\quad$ $5\cdot10^{-4} \ $ for (b), $\quad$ $2\cdot10^{-4} \ $ for (c), $\quad$ $5\cdot10^{-4} \ $ for (d).}
\end{equation}
 This is the dynamical counterpart of $\left|\widehat{a_+}(\lambda)\right| = 1$ in \Cref{thm:scattering}: at (most) energies in the gap the defect fully transmits.

Second, the lost mass is lost to the bulk, not to reflection. Each defect emits a circular wave at the moment the packet passes it (heavier when the turn is sharper). The fraction of $\|u(t)\|^2_{L^2}$ lying outside a strip of width $3$ around the interface at $t=24$ is
\begin{equation}\label{eq:numerics-bound2}
    0.078 \ \text{ for (a)}, \qquad 0.133 \ \text{ for (b)}, \qquad 0.104 \ \text{ for (c)}, \qquad 0.077 \ \text{ for (d)} .
\end{equation}
None of these four losses exceeds $\operatorname{erfc}(1) \simeq 0.157$, the fraction of the packet that sits outside the gap.

\end{document}